\documentclass[conference]{IEEEtran}
\let\labelindent\relax
\usepackage{enumitem}

\usepackage{cite}
\usepackage{amsmath,amssymb}
\usepackage{bm}
\usepackage{caption}
\usepackage{subcaption}

\usepackage{mathtools}
\usepackage{amsthm}
\usepackage{mdwmath}
\usepackage{amsmath}

\usepackage{mdwtab}
\usepackage[utf8]{inputenc}
\usepackage[english]{babel}
\usepackage{xcolor}

\usepackage{algorithm}
\usepackage{algpseudocode}
\usepackage{subcaption}

\usepackage{pgfplots}
\usetikzlibrary{matrix}
\usepgfplotslibrary{groupplots}
\pgfplotsset{compat=newest}
\usetikzlibrary{patterns}

\theoremstyle{plain}
\newtheorem{theorem}{Theorem}[section]
\theoremstyle{plain}
\newtheorem{lemma}{Lemma}[section]
\theoremstyle{definition}
\newtheorem{definition}{Definition}[section]
\definecolor{clr1}{RGB}{179, 179, 255}
\definecolor{clr2}{RGB}{77, 77, 255}
\definecolor{clr3}{RGB}{0, 0, 230}

\newcommand{\miracle}{{\sc MiRACLE}}
\newcommand{\subparagraph}{}
\usepackage[font=small,skip=0pt]{caption}
\ifCLASSINFOpdf
\else
\fi
\begin{document}
%
% paper title
% can use linebreaks \\ within to get better formatting as desired
\title{YODA: Enabling computationally intensive contracts on blockchains with Byzantine and Selfish nodes}
% author names and affiliations
% use a multiple column layout for up to three different
% affiliations
% \author{\IEEEauthorblockN{Sourav Das}
% \IEEEauthorblockA{Indian Institute of Technology Delhi\\
% souravdas1547@gmail.com}
% \and
% \IEEEauthorblockN{Vinay Joseph Ribeiro}
% \IEEEauthorblockA{Indian Institute of Technology Delhi\\
% vinay@cse.iitd.ac.in}
% \and
% \IEEEauthorblockN{Abhijeet Anand}
% \IEEEauthorblockA{Indian Institute of Technology Delhi\\
% abhijeetanand98765@gmail.com}
% }

% conference papers do not typically use \thanks and this command
% is locked out in conference mode. If really needed, such as for
% the acknowledgment of grants, issue a \IEEEoverridecommandlockouts
% after \documentclass

% for over three affiliations, or if they all won't fit within the width
% of the page, use this alternative format:
% 
\author{\IEEEauthorblockN{Sourav Das\IEEEauthorrefmark{1},
Vinay Joseph Ribeiro\IEEEauthorrefmark{2} and
Abhijeet Anand\IEEEauthorrefmark{3}}
\IEEEauthorblockA{
Department of Computer Science and Engineering,
Indian Institute of Technology Delhi, India\\
\{\IEEEauthorrefmark{1}souravdas1547, \IEEEauthorrefmark{3}abhijeetanand98765\}@gmail.com, \IEEEauthorrefmark{2}vinay@cse.iitd.ac.in}}
% use for special paper notices
%\IEEEspecialpapernotice{(Invited Paper)}
% \IEEEoverridecommandlockouts
% \makeatletter\def\@IEEEpubidpullup{9\baselineskip}\makeatother
% \IEEEpubid{\parbox{\columnwidth}{Permission to freely reproduce all or part
%     of this paper for noncommercial purposes is granted provided that
%     copies bear this notice and the full citation on the first
%     page. Reproduction for commercial purposes is strictly prohibited
%     without the prior written consent of the Internet Society, the
%     first-named author (for reproduction of an entire paper only), and
%     the author's employer if the paper was prepared within the scope
%     of employment.  \\
%     NDSS '16, 21-24 February 2016, San Diego, CA, USA\\
%     Copyright 2016 Internet Society, ISBN 1-891562-55-X\\
%     http://dx.doi.org/10.14722/ndss.2016.23xxx
% }
% \hspace{\columnsep}\makebox[\columnwidth]{}}

\IEEEoverridecommandlockouts
\makeatletter\def\@IEEEpubidpullup{6.5\baselineskip}\makeatother
\IEEEpubid{\parbox{\columnwidth}{
    Network and Distributed Systems Security (NDSS) Symposium 2019\\
    24-27 February 2019, San Diego, CA, USA\\
    ISBN 1-891562-55-X\\
    https://dx.doi.org/10.14722/ndss.2019.23142\\
    www.ndss-symposium.org
}
\hspace{\columnsep}\makebox[\columnwidth]{}}

% make the title area
\maketitle
\begin{abstract}
%\boldmath
One major shortcoming of permissionless blockchains such as
Bitcoin and Ethereum is that they are unsuitable for running
Computationally Intensive smart Contracts (CICs). This prevents such
blockchains from running Machine Learning algorithms, Zero-Knowledge proofs,
etc. which may need non-trivial computation.

In this paper, we present YODA, which is to the best of our knowledge
the first solution for efficient computation of CICs in permissionless
blockchains with guarantees for a threat model with both
Byzantine and selfish nodes. YODA selects one or more 
execution sets (ES) via Sortition to execute a particular CIC off-chain. 
One key innovation is the MultI-Round Adaptive Consensus using 
Likelihood Estimation (\miracle) algorithm based on sequential 
hypothesis testing. \miracle\  allows the execution sets to be small thus 
making YODA efficient while ensuring correct CIC execution with high 
probability. 
It adapts the number of ES sets automatically depending on the
concentration of Byzantine nodes in the system and is optimal in terms
of the expected number of ES sets used in certain scenarios.
Through a suite of economic incentives and technical mechanisms such
as the novel Randomness Inserted Contract Execution (RICE) algorithm,
we force selfish nodes to behave honestly. We also prove that the
honest behavior of selfish nodes is an approximate Nash
Equilibrium. We present the system design and details of YODA and
prove the security properties of \miracle\ and RICE.
Our prototype implementation built on top of Ethereum demonstrates the
ability of YODA to run CICs with orders of magnitude higher gas per
unit time as well as total gas requirements than Ethereum currently
supports. It also demonstrates the low overheads of RICE.

\end{abstract}
\section{Introduction}
\label{sec:introduction}

Permissionless blockchain protocols, which originated with
Bitcoin~\cite{nakamoto2008bitcoin}, allow an arbitrarily large network
of {\em miners} connected via a peer-to-peer overlay network to
agree on the state of a shared ledger. 
% New transactions broadcast over
% the P2P network are periodically coalesced into blocks and added to
% the blockchain through a process called mining. The transactions and
% their order determine the state of the ledger.
More recent blockchains extend the shared ledger concept to
allow programs called {\em smart contracts} to run on them~\cite{buterin2014next,
szabo1996smart}. Smart contracts maintain state 
that can be modified by transactions. One of the major shortcomings 
of these blockchains is that they are unsuitable for smart contracts 
which require non-trivial computation for execution~\cite{
croman2016scaling}. We call such smart contracts {\em Computationally 
Intensive Contracts} (CIC). CICs can potentially run
intensive machine learning algorithms~\cite{wang2014learning},
zero-knowledge proofs~\cite{ben2014succinct,eberhardtzokrates}~etc.

One reason for this shortcoming is that every transaction is executed 
{\em on-chain}, that is by {\em all} miners, and this computation must 
be paid using the transaction fee. Hence CIC transactions require very 
high transaction fees.\footnote{Transaction verification is to some 
extent subsidized by mining fees.} 
A second reason is the Verifier's Dilemma~\cite{luu2015demystifying}. 
A miner must normally start mining a new block on an existing block only 
after verifying all its transactions. If the time taken to verify these 
transactions is non-trivial, it delays the start of 
the mining process thereby reducing the chances of the miner creating 
the next block. Skipping the verification step will save time but at 
the risk of mining on an invalid block, thereby leaving a rational 
miner in a dilemma of whether to verify transactions or not.

One mechanism to side-step the Verifier's Dilemma is to break a 
computationally-heavy transaction into multiple light-weight 
transactions and spread these out over multiple blocks~\cite{luu2015demystifying}.
This mechanism has several shortcomings. 
First, the total fees of these transactions
may be prohibitively high. Second, how to split a single general
purpose transaction into many while ensuring the same resulting ledger
state is not obvious. Third, the
number of blocks over which the light-weight transactions are spread
out grows linearly with the size of the total computation.

Another approach is to execute smart contracts {\em off-chain}, i.e.
by only a subset of {\em nodes}\footnote{For clarity, we  use the 
term {\em node} for an entity performing off-chain computation
and the term miner for an entity performing all on-chain computation.}
to cut down transaction fees and avoid the Verifiers 
Dilemma~\cite{teutsch2017scalable,kalodner2018arbitrum}.
% thereby cutting down on mining fees and side-stepping the Verifier's 
% Dilemma~\cite{teutsch2017scalable,kalodner2018arbitrum}. 
The off-chain methods proposed so far, however, work under the limited
threat model of nodes being rational and honest but not Byzantine. Moreover, they 
require on-chain computation of a part of a CIC to determine its correct 
solution in some cases. 
Note that off-chain CIC computation is not the same as achieving 
consensus about blocks using shards~\cite{luu2016secure,
kokoris2018omniledger,al2017chainspace,gilad2017algorand,zamani2018rapidchain}. 
Blocks can in general take many valid values and are computationally easy to
verify unlike CIC solutions which have only one correct value and are
% computationally 
expensive to verify.

% \vinay{check above}\sourav{Miners being rational or honest.}

%Byzantine,
%Altruistic, Rational (BAR) models which consider both Byzantine
%and Rational entities are more realistic
%challenging to analyse and
%more realistic than either of these two
%models~\cite{aiyer2005bar}. (3) Can such a system for CIC computation
%work under a BAR model?

% \vinay{Removing the questions we had raised about small subset etc. since
%   we say that it has been solved by Arbitrum etc.}

%To achieve
%scalability in computation we ask ourselves the following questions:
%(1) Is it possible to design a permissionless blockchain system in
%which only a {\em small subset} of nodes execute a CIC but which
%provides the same security guarantees of Ethereum or Bitcoin? (2) Is
%it possible to design a blockchain system which decouples transaction
%execution from existing Proof-of-Work (PoW) consensus, thereby
%side-stepping the Verifier's Dilemma? Further, most research in
%permissionless blockchains has focused on threat models which either
%have only rational nodes or ones which have Byzantine and altruistic
%(honest) nodes but no rational nodes~\cite{sapirshtein2016optimal,
%  pass2017fruitchains,gilad2017algorand,pass2017sleepy,luu2016secure}.

\noindent{\bf Our Goal.} Our goal is to design a mechanism for
off-chain CIC execution with the following properties.
\begin{enumerate}
  \item {\em BAR Threat model.} It should work under a Byzantine,
    Altruistic, Rational (BAR) model which considers both Byzantine and 
    Rational entities. BAR models are more realistic and challenging 
    to analyze than threat models which consider only one of Byzantine 
    or rational entities~\cite{aiyer2005bar}. 
  \item {\em Adaptive to Byzantine fraction.} It should make fewer nodes
    perform off-chain computation if the fraction of Byzantine nodes is 
    smaller.
  \item {\em Scalability}: CICs are never executed or verified on-chain
   either fully or partially. Further, the number of CICs that can be 
   executed in parallel must scale with increasing number of nodes in the 
   system.
% \item {\em On-chain Equivalence.} All off-chain CIC execution
%   must terminate correctly and the state of any CIC must be easily obtained
%   from information on the blockchain, similar to on-chain execution of
%   non-CIC smart contracts
%   in existing blockchains.\vinay{commented out the n/2 statement. We could mention it in the BAR model above but we have not said what SP is so far.}
%  , further their the post execution state are available on-chain as long as adversary does not control more than $n/2$ nodes in the system where remaining nodes are selfish. 
% \vinay{We can remove the following as we have not said what SP is.
% 	\item {\em Scalability.} CICs are never executed or verified on-chain either fully or partially. Further, the number of ITs that can be executed in YODA must scale with increasing number of nodes in SP. 
% }
  \item{\em Fair and timely reward.} As CIC execution is expensive, all 
  nodes performing off-chain computation correctly must be compensated 
  fairly and in a timely manner.  
\end{enumerate}

\noindent{\bf Our Approach.}
In this paper we present YODA, which is to the best of our knowledge
the first solution for efficient computation of CICs in permissionless
blockchains which gives security guarantees in a BAR threat model. The
threat model allows at most a fraction $f_{max}<0.5$ of Byzantine
nodes in the overall system and the remaining can be {\em
  quasi-honest}. Note that the actual fraction $f$ of Byzantine nodes
is unknown a priori and can be anywhere between $0$ and $f_{max}$.
Although YODA is designed for the worst case ($f=f_{max}$), it
adapts to smaller values of $f$, by evaluating CICs more efficiently.

Quasi-honest nodes are selfish nodes which seek to
maximize their utility by skipping CIC computation using information
about its solutions which may already be published on the blockchain
by other nodes. We call this a {\em free-loading attack}. They may
also try to collude with each other to reduce their computation. More
details about quasi-honest nodes are given in~\textsection
\ref{sub:threat assumption}. YODA is robust to DoS attacks, Sybil
attacks, and ensures timely payouts to all who execute a CIC.

YODA's modus operandi is to make only small sets of randomly selected
nodes called Execution Sets (ES) compute the CICs.  ES nodes submit
their solutions, or just a small digest of them, on the blockchain as
transactions. YODA then study the counts of various solutions
submitted in order to identify the correct solution from among them.
While a small ES improves system efficiency, it can occasionally be
dominated by Byzantine nodes which may form a majority and submit
incorrect solutions. Hence, a simple majority decision does not work
even in a setting with only honest and Byzantine nodes.

To determine the correct CIC solution, YODA uses a novel {\em
  MultI-Round Adaptive Consensus using Likelihood Estimation}
(\miracle) algorithm. In \miracle, miners compute the {\em likelihood}
of each received digest which primarily depends on the counts of 
different digests and the fraction $f$ of Byzantine nodes. 
If the likelihood of any digest crosses a particular threshold, \miracle\ 
declares its corresponding solution as the correct one. 
Otherwise, it iteratively selects additional ES sets until the likelihood
of a digest crosses required threshold. We call the 
selection of each such ES a {\em round}.\footnote{Rounds are 
different from block-generation epochs and are specific to CICs. A 
round may span multiple blocks.}
\miracle\ is adaptive, that is the expected number of rounds
it takes to terminate is smaller the smaller $f$ is.
\miracle\ guarantees selection of the correct digest with probability 
at least $1-\beta$ for a design parameter $\beta$. Moreover, for the 
special case of $f=f_{max}$, if all Byzantine node submit the same 
incorrect digest, \miracle\ optimally minimizes the expected number of 
rounds. Interestingly, the strategy for Byzantine nodes
to make \miracle\ accept an incorrect solution with highest probability
is to submit the same incorrect solution~(refer~\ref{sub:miracle analysis}).

% \miracle\ guarantees that the correct solution is chosen with
% probability set by a design parameter in the worst case of
% $f=f_{max}$.  Moreover, \miracle\ is optimal in that it minimizes the
% expected number of rounds in the special case that all Byzantine nodes
% submit the {\em same} incorrect solution and
% $f=f_{max}$. Interestingly, the strategy for Byzantine nodes
% to make \miracle\ accept an incorrect solution with highest probability
% is to submit the same
% incorrect solution.

This analysis for \miracle, however, assumes that all quasi-honest nodes
submit correct solution. Since \miracle\ itself does not enforce honest
behaviour, other mechanisms are necessary to make quasi-honest nodes
submit correct solutions. Without additional mechanisms, a
quasi-honest node may be tempted to free-load on solutions already
submitted in earlier rounds, 
%for example by choosing a digest with
%highest likelihood so far, 
thus saving on computational power. In case quasi-honest
nodes free-load on incorrect solutions, \miracle\
has a higher probability of terminating with an incorrect solution.
%\sourav{Further, free-loading attack, increases the probability 
%of YODA accepting an incorrect solution. Specifically, in case the digest
%with highest likelihood in the previous rounds is incorrect.
%This might happen as ES sizes are chosen small which can sometimes
%be dominated by an adversary}
%\sourav{{\bf We have to little careful with the above statement
%and mention that adversarial dominance does not imply an incorrect
%solution will get accepted. This might be little
%confusing to people reading for the first time.}}
%In a BAR threat model, \miracle\ by itself does not force quasi-honest
%nodes to behave honestly. In fact, a {\em free-loading attack} by
%quasi-honest nodes is a real possibility. Here quasi-honest nodes in
%an ES of one round may simply replay the digest with highest
%log-likelihood of previous submissions in earlier rounds, thereby
%guessing the correct digest with large probability and also saving on
%heavy CIC computation.

To mitigate the free-loading attack of quasi-Honest nodes, we design
the Randomness Inserted Contract Execution (RICE), an efficient
procedure to change the digest from one round to the next. We achieve 
this by making the digest dependent on a set of pseudo-randomly chosen
intermediate states of a CIC execution. This ensures, that despite digests 
changing from one round to the next, the miners running \miracle\ 
are able to map digests from different rounds to the same CIC state 
they represent. We prove that 
% no matter the size of the CIC computation, 
RICE adds little computational overhead to CIC execution.
To be precise, if $T$ denotes the total computation for a transaction 
execution without RICE, then RICE adds computation overhead of  
$O((log_2 T)^2)$. In the presence of free-loading attacks, we show via 
a game theoretic analysis that honest behavior from all quasi-honest 
nodes is an $\epsilon-$Nash equilibrium with $\epsilon\ge 0$.

We have implemented YODA with \miracle\ and RICE, in Ethereum 
as a proof-of-concept and provide many experimental results supporting 
our theoretical claims.
%using the geth Client interface
%version \texttt{1.8.7}. The implementation includes \miracle, RICE, and
%most other features of YODA.
%Using a testbed consisting of 8 physical
%machines emulating 1600 nodes, we
%and compare its performance with
%that of Ethereum in terms of gas usage.

%the gas usage possible per unit time as
%well as the total gas usage of contracts.  We also study the number of
%rounds for \miracle\ to converge for different design parameters. We
%show how \miracle\ automatically reduces the number of rounds if the
%fraction of Byzantine nodes is less than the worst case design
%scenario. We also show that RICE hardly adds any overhead to the CIC
%computation.

%\noindent{\bf Paper Organization}. In Section~\ref{sec:system design} we present a system overview of YODA and mention requirements for off-chain execution. This is followed by the Threat Model and an overview of Challenges in Section~\ref{sec:threat assumption}. In~\textsection \ref{sec:miracle} we present the \miracle\ algorithm followed by a detailed description of RICE in Section \ref{sec:rice}. We then put all ingredients together to describe the YODA protocol in Section \ref{sec:enabling cic} followed by a Security Analysis in Section \ref{sec:analysis}. We then present details of our Ethereum-based implementation along with experimental results in Section \ref{sec:eval}. Section \ref{sec:related} describes related work. We conclude with a discussion in Section \ref{conclusion}.

% !TEX root = ../main.tex
\section{Theat Model, Assumptions and Challenges}
\label{sec:system model}
In YODA, a blockchain is an append-only distributed ledger consisting
of data elements called blocks. A blockchain starts with a pre-defined
genesis block. Every subsequent block contains a hash pointer to the
previous block resulting in a structure resembling a chain. The
blockchain contains accounts with balances, smart contracts, and
transactions. A transaction is a signed
message broadcast by a account owner which can be included in a block provided
it satisfies certain validity rules. For example, transactions
modifying an account balance must be signed by the corresponding private key
to be valid. YODA assumes that the underlying blockchain provides
guarantees about its Safety and Availability. {\em Safety} means that
all smart contract codes are executed correctly on-chain, and {\em
  availability} means that all transactions sent to the blockchain
get included in it within bounded delay and cannot be removed thereafter.

We refer to any entity performing off-chain CIC execution in 
YODA as a {\em node}. We call the set consisting of all nodes the
{\em Stake Pool} (SP).\footnote{The choice of the name
  will be discussed when we discuss blockchain specifics.}
Without loss of generality each node in SP controls an
account in the ledger with its private key. The account itself is
identified by the public key. 
We assume that the network is synchronous, i.e.,
transactions broadcast by nodes get delivered within a known bounded
delay. However, unlike~\cite{luu2016secure} we do not assume the
existence of an overlay network among nodes. Also, we do not assume
the presence of a secure broadcast channel or a PKI system. We
abstract the source of randomness required for RICE to a function ${\sf
RandomGen()}$ (given in~\textsection \ref{sec:enabling
cic}) which can be accessed by all nodes in YODA. This can be built
as a part of YODA or as an external source using techniques from~\cite{
syta2017scalable,luu2016secure,gilad2017algorand}.

For the rest of the paper, unless otherwise stated, if some event has
{\em negligible} probability, it means it happens with probability at
most $O(1/2^{\lambda})$ for some security parameter
$\lambda$. Any event whose complement occurs with negligible
probability is said to occur with high probability or $w.h.p$.

\subsection{ Threat Model and Assumptions.}
\label{sub:threat assumption}
Systems like permissionless blockchains cannot be assumed to have all
honest nodes. They rely heavily on incentives and the rationality of
nodes in order to work correctly.  Rational nodes are those which seek
to maximize their utility.  However, assuming that all nodes are
rational is not practical either. Real systems may contain Byzantine
nodes, that is those which do not care about their returns.

We consider two kinds of nodes: {\it Byzantine} and {\it
 quasi-Honest}. Byzantine nodes are controlled by an {\it adversary}
and these nodes can deviate arbitrarily from the YODA protocol. 
The adversary can make all Byzantine nodes collude with perfect
clock synchrony. They can add or drop messages arbitrarily and not
execute CICs correctly. We assume that at most $f_{max} < \frac{1}{2}$
fraction of nodes in SP are Byzantine. 
% The adversary can  arbitrarily
% select these nodes at the start of each round. However, during
% the progress of a single round~(see \textsection
% \ref{sec:miracle}) the adversary cannot compromise more nodes. 
Additionally the adversary has  state 
information of the CIC from all previous rounds and can successfully communicate 
this (potentially false) state information  about previous rounds to 
any node in SP.
% the system before the start of a round. 
However, we assume 
cryptographic primitives are computationally secure.

Modeling rational nodes in these systems, taking into
account all possible means of profits, costs, and attacks is
non-trivial and is beyond the scope of the paper. However to bring 
our model close to reality we work with quasi-honest nodes which
deviate from the protocol in the following manner.

\noindent{\bf Quasi-Honest.}  Quasi-honest nodes will skip execution of a CIC
either completely or partially, for example by not executing some of
its instructions, if and only if the expected reward in doing so is
more than that for executing the transaction faithfully. They do not
share information with any other node if that information
can lead to reduction of their reward. 
They are {\it conservative} when estimating the potential impact of 
Byzantine adversaries in the system, i.e. a quasi-honest node while 
computing its utility assumes that the Byzantine adversary acts 
towards minimizing its (quasi-honest node's) rewards~\cite{aiyer2005bar}.

Quasi-honest nodes may skip computation using one of two methods. The
first is ``free-loading'' where they attempt to guess the correct
state of a CIC after execution of a  transaction from the information of the
corresponding transaction already published on the blockchain. 
Free-loading also includes the case where a quasi-honest node tries to
guess the state when an adversary presents  the
pre-image of hashes among this information already published on the blockchain.
% \vinay{removing the comment about RICE which is too complicated to
%   understand here.}
%, or from 
%(ii) state information that can be derived or verified by re-executing
%the RICE algorithm using the same information in (i).
%Quasi-honest nodes may skip computation using one of two methods. The
%first is ``free-loading'' where they attempt of identify the correct
%state of a CIC after a transaction by using (i) the information of 
%the corresponding transaction already published on the blockchain, 
%or (ii) state information that can be derived or verified by re-executing 
%the RICE algorithm using the same information in (i).

The second is by colluding with other ES nodes of the same round to
submit an identical CIC solution without evaluating it. A
quasi-honest node only colludes with nodes whose membership in the ES
it can verify. YODA has checks (refer~\textsection~\ref{sub:commitment}
) which prevent nodes from directly proving their ES membership. 
Hence nodes must use Zero-Knowledge-Proof techniques
like zk-SNARK~\cite{ben2014succinct} to establish their membership in
ES. YODA allows use of smart-contracts as shown in~\cite{juels2016ring} 
to establish rules of collusion. However we
assume that a quasi-Honest node does not know for sure if the node it
is colluding with is quasi-honest or Byzantine. Additionally, both
free-loading and collusion have costs associated with them.
These cost are due to processing of information available on the blockchain or 
received from peers, producing and verifying zk-Proofs, 
bandwidth and computation costs etc. In case neither
free-loading nor collusion gives a better expected reward than
executing CICs correctly, a quasi-honest node will execute the CIC
correctly.
% \vinay{Are we allowing zk-proofs here?}
% \sourav{I think this explanation is fine. A node can use zk-proofs
% to convince others about its selection in ES.}

\subsection{Challenges}
\label{sub:challenges}
Enabling off-chain execution of CICs in the presence of a Byzantine
adversary is fraught with many challenges. Allowing non-Byzantine
nodes to deviate from the protocol makes the problem more interesting 
and even more challenging. Apart from recently studied challenges like
preventing {\it Sybils}~\cite{nakamoto2008bitcoin,zamani2018rapidchain} and
generating an unbiased source of randomness in the distributed 
setting~\cite{gilad2017algorand,luu2016secure,
syta2017scalable}, our system must tackle the following challenges:
\begin{itemize}
	\item to prevent quasi-honest nodes from {\em Free-loading and collusion}.
	\item since the size of any ES is small, an ES becomes vulnerable to 
			{\it Lower cost DoS Attacks} than a DoS attack on the set of all nodes 
			taken together.
	\item to provide guarantees of correctness without requiring re-execution 
			of any part of the CIC on-chain.
\end{itemize}

% The first challenge is to prevent
% quasi-honest nodes from {\em Free-loading and collusion}.
% The second challenge is that since the size of any ES is small, an ES 
% becomes vulnerable to {\it Lower cost DoS Attacks} than a
% DoS attack on the set of all nodes taken together. 
% A third challenge is to provide guarantees of correctness without
% requiring re-execution of any part of the CIC on-chain.  

% \vinay{Commented out symmetric/asymmetric}
%\sourav{We can shift the following content to either Introduction or 
%Related work.}
%Existing on-chain verification of off-chain computation methods fall into two
%categories: Symmetric and Asymmetric. In symmetric methods, a contract
%is re-executed on-chain by miners and is hence its computation is
%%%y generating a proof-of-execution, possibly interactive~\cite{
%teutsch2017scalable}, or via non-interactive methods using
%zk-SNARKs~\cite{ben2014succinct}. These techniques are associated with
%non-deterministic overheads (depending upon the CIC)
%(eg. Truebit~\cite{teutsch2017scalable}) or are too computationally
%expensive for general purpose hardware. We wish to design a
%verification scheme that adds a small constant computational overhead
%for on-chain verification, independent of the CIC.

 %We are the first to take
%this approach. Also, since CICs are never executed on-chain, producing
%both efficient off-chain execution and providing guarantees of
%correctness is challenging.

% !TEX root = ../main.tex
\section{\miracle: Multi-Round Adaptive Consensus using Likelihood Estimation}
\label{sec:miracle}
In this section we describe \miracle\ as an abstract consensus protocol 
and later get into its blockchain specifics.
%The nodes
%can by Byzantine or honest. The maximum Byzantine fraction is $f_{max}$ The goal is %to make all honest nodes  achieve consensus
%regarding the output of a function $\Psi$ for a given input $x$ by having
%only a small subset of SP evaluate the function..

%The goal of \miracle\ is
%to be efficient, by making only a few nodes to compute $\Psi$ and at
%the same time automatically adapt to the fraction $f$ of Byzantine
%nodes in SP by terminating faster and fewer nodes to compute
%$\Psi(x)$. We prove that \miracle\ is optimal in the expected number
%of rounds if the Byzantine fraction of nodes in SP equals
%$f_{max}$. We first describe \miracle\ as a abstract consensus
%algorithm and later concretize for our off-chain execution.
\noindent{\bf Problem Definitions.}  Let $\Psi$ be a deterministic
function that when given arbitrary input $x$ produces output $y$. We
denote this as $y \leftarrow \Psi(x)$.\footnote{We use $\leftarrow$ 
for function executions, with the function and its inputs on its 
right and the returned value on its left.} Let SP contain at most
$f_{max}$ fraction of Byzantine nodes. All other nodes are honest,
i.e. they strictly adhere to the protocol. Let $n_i$ be a node in SP
where $i=1,2,\ldots, |SP|$.  Let $\Psi_i$ be the function $n_i$
executes when asked to execute $\Psi$ and let $y_i \leftarrow
\Psi_i(x)$ be the corresponding result. For all honest
nodes, clearly $\Psi_i = \Psi$.

Our goal is to achieve consensus on the true value of $\Psi(x)$ by 
making only one or more small randomly chosen subsets called Execution 
Sets (ES) of nodes evaluate $\Psi(x)$.  Further,
nodes $n_i\ \forall i$ after executing $\Psi_i(x)$ broadcast a 
{\em digest} of $y_i$, say ${\sf hash}(y_i)$ to  all other nodes. 
\miracle\ proceeds in {\em rounds} where in each round a new ES is  
selected. 
We require \miracle\ to correctly reach consensus on $\Psi(x)$ with 
probability greater than $1-\beta$ for any given  user-specified 
parameter $\beta$, given $f_{max}$ and $\mathbb{E}[|ES|]$, while 
minimizing the expected number of rounds to terminate. 
Formally, \miracle\ must guarantee the following properties.
\begin{itemize}
  \item {\bf (Termination)} For any $f_{max} < 1/2$,
    \miracle\ must terminate within a finite number of
    rounds.
    % \vinay{We have not discussed termination in the paper. Also
    %   it should terminate with probability 1 and not with 
    % probability $1-\beta$ as you stated.}
  \item {\bf (Agreement)} All honest nodes in SP, agree on the result 
  that \miracle\ returns on terminating.
  \item {\bf (Validity)} \miracle\ must achieve consensus on the true
    value of $\Psi(x)$ with probability $1-\beta$.
  \item {\bf (Efficiency)} When the fraction of Byzantine nodes is
    $f_{max}$ and given a particular $\mathbb{E}[|ES|]$, 
    \miracle\ must terminate in the optimal number of rounds. 
    Further, for any given $f\le f_{max}$, if $N_{f}$ denotes the 
    expected number of nodes performing off-chain execution then 
    $N_{f} \le N_{f_{max}}$  
\end{itemize}

\subsection{Overview and Simplistic algorithms}
\label{sub:simplistic algorithms}
We motivate \miracle\ by describing two simplistic algorithms for
achieving consensus regarding $\Psi(x)$. 
%\sourav{The following few statements are specific to NS1 and NS2.
%For general \miracle\ this is not required. Also our explanation of 
%\miracle\ does not rely on the below. We can re-phrase this so that
%we can present \miracle\ in the very general sense.}
In these two algorithms each node in SP belongs to a particular ES 
with probability $q$ independent of other nodes, thus
$\mathbb{E}[|ES|]=q |SP|$. Note that \miracle\ in general need not
have the same $\mathbb{E}[|ES|]$ in every round, although in this
paper we present a version which does. Studying other possible
 \miracle\
algorithms is part of our future work.

\noindent{\bf Naive Solution 1 (NS1):} Suppose we use a single subset
ES from SP to compute $\Psi(x)$. If more than 50\% of nodes in the ES
publish the same execution result then this is chosen as
$\Psi(x)$. One shortcoming of this scheme is that for lower $\beta$~(more~security), 
the size of ES must be a large fraction of SP. A second shortcoming is
that if the actual fraction of Byzantine nodes $f$ is much smaller
than $f_{max}$ then we end up using an ES much larger than
required. For example, with $\beta = 10^{-20}$ as the error,
starting with $|SP|=1600$ and $f_{max}=0.35$, NS1 will always pick 
$|ES|\approx 900$ independent of $f$.

\noindent{\bf Naive Solution 2 (NS2):} In this solution we relax the
requirement of achieving consensus in one round. If in an ES, the 
fraction of nodes submitting a particular solution exceeds some 
threshold then we terminate with that solution. This threshold should 
be high enough to ensure the correct solution $w.h.p$. In NS1 for 
example, the threshold is $1/2$. In general, the smaller $q$ is the 
larger will the  threshold be. If we do not reach consensus then a new 
round is triggered.

The advantage is that we can use an ES in each round of size smaller
than the ES used in NS1.  In NS2, in certain instances such as when
$f=0$, a single round may still be sufficient to reach consensus.
One shortcoming is that the number of rounds to terminate can
be large because NS2 does not optimally combine the results of all
rounds in order to reach consensus. Results of one round are forgotten
in future rounds.

\subsection{Design and Algorithm}
\label{sub:miracle design}
In \miracle, we employ the multi-round strategy of NS2 to achieve
gains in case $f\ll f_{max}$. In contrast to NS2, each round uses all
hitherto published results to decide whether to terminate or not. For
a given $\Psi(x)$, let $d_1,d_2,...,d_m$ be the $m$ unique digest
values broadcast up to and including the $i^{th}$ round. Let $c_{k,i}$
denote the number of times $d_k$ is repeated in the $i^{th}$ round.
% and $\boldmath{C}_{k,i}$ denote the corresponding random variable. 
Let $C_i$ denote the total number of submissions (ES nodes) in the
$i^{th}$ round, i.e. $C_i = \sum_{k=1}^{m}{c_{k,i}}$. 

The problem we are addressing is to decide among one of may solutions
broadcast. We present a novel model of this problem as a multiple
hypothesis testing problem where we have one hypothesis for each 
solution submitted and the test must decide which hypothesis is true.

\noindent{\bf Primer on Hypothesis Testing.}
For the reader unfamiliar with Hypothesis testing, we now describe
a standard example.  Consider a communication system in which 
a source is transmitting one symbol selected from a known small master set to
a receiver over a noisy channel. In the simplest case, only two
symbols are allowed, one each for communicating bit 0 and bit 1.
The receiver's task is to decide
which symbol (and hence which corresponding bit(s)) was transmitted given the observation. To solve the
problem, one proposes a hypothesis for each potential symbol which
claims that the corresponding symbol was transmitted. The goal is to
determine which hypothesis is true. To do so, the receiver computes
the probability of the observation conditioned on every hypothesis
being true. Only if one of these probabilities is much larger than the others can
one say with confidence that the corresponding hypothesis is true with
high probability.

One of our novel contributions in \miracle\ is to formulate the
problem of determining the true $\Psi(x)$ as a hypothesis testing
problem. This is not obvious because traditional
hypothesis tests are designed to handle real-world phenomena such as
signals in noise. In our problem we have an intelligent adversary
which is hard to model as there is no restriction on what solution it
can submit. It can submit any of $2^{n}$ digests if the digest is $n$
bits long. Hence unlike the communication problem described above,
there is no small master set of potential correct solutions known a
priori to YODA.

However, 
in the worst case when the fraction of Byzantine nodes is maximum,
i.e. $f=f_{max}$, we do have a probability
distribution on the {\em total number of Byzantine nodes} in an
ES. Similarly we have a probability distribution of the total number
of quasi-honest nodes in an ES. These probability distributions are
sufficient for us to compute a likelihood and perform a hypothesis
test. In the case the adversary submits only a single incorrect
solution, \miracle\ is optimal in the number of rounds it takes to
converge. However, if it submits many solutions then our assumed
distributions for different hypotheses are not exact. Fortunately, if
the adversary submits more than one solution, it is to his own
detriment as \miracle\ will converge to the correct solution with
higher probability than if it submitted only a single solution.

\miracle\  uses multiple parallel Sequential Probability Ratio Tests (SPRT) 
to choose the correct solution~\cite{wald1973sequential} whose details are given next.

\noindent{\bf \miracle\ as Parallel SPRT:} We model the problem as $m$ simultaneous two-hypotheses Sequential
Probability Ratio Tests (SPRT)~\cite{wald1973sequential}. The $k^{\rm
  th}$ SPRT is given by:
\begin{itemize}
  \item Null Hypothesis, $\mathcal{H}_k: d_k$ is the solution.
  \item Alternative Hypothesis, $ \mathcal{H}^{*}_k: d_k $ is not the solution.
\end{itemize}

The log-likelihood ratio is defined as the log of the ratio of
probabilities of the observations ($c_{k,i}$) conditioned on the two
hypotheses.  We approximate this log-likelihood ratio after $i$ rounds
by a quantity we loosely call the {\em likelihood}. We  denote it by
$L_{k,i}$, and proceed as follows. We give a formula for the
likelihood subsequently. For appropriately chosen threshold
$\mathbb{T}$, in round $i$ we perform
\begin{itemize}
  \item If $L_{k,i}\le \mathbb{T}$, make no decision.
  \item If $L_{k,i}> \mathbb{T}$, decide in favor of $\mathcal{H}_k$.
\end{itemize}

When any one SPRT, say the $k^{\rm th}$, terminates in favor of its
Null Hypothesis $\mathcal{H}_k$, we halt all other SPRTs and declare
$d_k$ as the digest. If no SPRT terminates, we proceed to the next
round.
%Trivially this general case reduces to the two-hypothesis case
%described earlier in the event that there are only two digests $d_1$
%and $d_2$ published on the blockchain.
Algorithm~\ref{algo:miracle}
demonstrates the general \miracle\ algorithm for any given $L_{k,i}$
and $\mathbb{T}$. 
% We prove several results pertaining to this parallel
% SPRT in~\textsection \ref{sub:miracle analysis}.

\begin{algorithm}[h!]
\caption{\miracle}\label{algo:miracle}
\begin{algorithmic}[1]
      \State $i\gets 1$
      \While{$L_{k,i} \leq \mathbb{T} ~\forall k$}
          \State $i \leftarrow i+1$
        \State Pick next ES to execute $\Psi(x)$
      \EndWhile
      \State \textbf{declare} $d_{k'}$ to be correct where $L_{k',i}>\mathbb{T}$ 
  \end{algorithmic}
\end{algorithm}

\noindent{\bf \miracle\ and YODA.}
We now present our specific choices for the likelihood $L_{k,i}$ and threshold
$T$ which we use in YODA. \miracle\ can in general use other choices
for the same quantities, and such generalisations are part of future work.

Recall that nodes are selected randomly and with the same probability 
$q$ for any ES. We set
\begin{equation}
    L_{k,i} = \sum_{j=1}^{i}\left(c_{k,j}^2 - (C_j - c_{k,j})^2\right) = \sum_{j=1}^{i} \left((2c_{k,j} - C_j)C_j\right)
\end{equation}
and the threshold to:
\begin{equation}
\label{equ:T}
    \mathbb{T} = \left(\ln\frac{1-\beta}{\beta}\right)\frac{2q(1-q)M(1-f_{max})f_{max}}{(1-f_{max})-f_{max}}.
\end{equation}
The above choice of $L_{k,i}$ is indeed the log-likelihood ratio when the
adversary submits a single incorrect solution in all rounds, assuming
a Gaussian distribution for the number of quasi-honest and
Byzantine nodes in any ES. Under the same assumptions, in 
Section~\ref{sub:miracle analysis} we describe how this choice 
of threshold gives an optimal solution in terms of number
of rounds to terminate along with required security.

% !Tex root=../main.tex
\section{RICE: Randomness Inserted Contract Execution}
\label{sec:rice}
\miracle\ by itself does not force quasi-honest nodes to behave
honestly. In fact, a {\em free-loading attack} by quasi-honest node
is a real possibility. Here quasi-honest nodes in an ES of one round
may simply replay the digest with highest likelihood of previous
submissions in earlier rounds. Even though this enables a 
quasi-honest node to save on heavy CIC computation, this attack can
make increase the probability
of accepting an incorrect solution to larger than $\beta$. Specifically, in 
scenarios where the corresponding digest is an incorrect solution, as 
an adversary can sometime dominate a large fraction in ES, free-loading 
can lead to acceptance of an incorrect solution. 

To address this in this section we describe Randomness Inserted Contract
Execution (RICE), a procedure to pseudo-randomly change the digest of
$\Psi(x)$ from one round to the next to mitigate the free-loading
problem. Other attacks, such as collusion of quasi-honest nodes within
the same ES and copying digests submitted by nodes in the same round
are addressed in~\textsection\ref{sec:enabling cic}.

So far we have looked at $\Psi$
as a very abstract function without describing any of its details.
We now formally define the semantics of $\Psi$ required
to understand RICE.

%A naive attempt to achieve this is by associating $\Psi$ during its runtime with the private key of the quasi-Honest node to produce a unique digest for each node in a way whose correctness verifiable using its public key. All static association strategy to produce the digest will trivially fail as the adversary can release the complete state generating the dominant digest.
%This approach must also make the output of $\Psi$ unique to the node's public/private key {\bf SOURAV: If you recall, Truebit is propose an approach similar to this}. Issue is, the mapping of digest with public key of the node makes the task of verification and correlation among digest submitted by another nodes in the round computationally expensive. An adversary can also potentially use this to launch a DDoS attack in the system by submitting lots of incorrect digests that does not contribute in meeting Termination condition of MiRACLE.  

\subsection{Design of RICE}
\label{sub:rice design}
\noindent{\bf Setup.} 
We assume $\Psi$ 
to be a stateful function similar to a smart contract. 
Let $\sigma$ be the state on which $\Psi$
operates by taking an input $x$. The output of $\Psi(\sigma, x)$ is
the modified state $\sigma^*$. Call $root(\sigma)$ (or simply $root$)
a unique digest of $\sigma$. For example, $root$ can represent the
root of a Merkle tree where leaves of the tree correspond to the
contents of $\sigma$.

Let $j$ ($\ge 1$) be the \miracle\ round number. We wish to generate a pseudorandom 
digest in each round. At the same time, to compute likelihoods, we
must be able to map digests across different rounds to each other. 
To solve the paradox, we generate a digest $(seed^{(j,.)},root)$, 
where $seed^{(j,.)}$ is a pseudorandom number which changes from 
one round to the next. The $seed$ is initialized as:
\begin{equation}
  seed^{(j,0)} \leftarrow
  \begin{cases}
  {\sf RandomGen}() & \text{if } j=1 \\
  {\sf hash}(seed^{(j-1, 0)}) & \text{otherwise}
  \end{cases}
\end{equation}

% \vinay{We use $root$ and ${\sf root}$. We should keep one.}
\noindent{\bf Array Model for RICE.}  Consider an execution
model in which all machine level instructions that $\Psi(\sigma,x)$
executes are stored in an imaginary ``instruction array'', that is the
$i^{\rm th}$ instruction executed is stored in the $i^{\rm th}$ array
position. RICE then interrupts execution\footnote{Blockchains such as
  Ethereum count gas used after each instruction. Hence additional
  interrupts are not required for Ethereum-like blockchains.} of
$\Psi(\sigma,x)$ at certain intermediate indices of the array where
state of the CIC is $\sigma'$ and
updates the $seed$ as follows:
\begin{equation}
  seed^{(j,l+1)} \leftarrow {\sf hash}(seed^{(j,l)}||root(\sigma'))
\label{eq:seedupdate}
\end{equation}
By choosing these different indices pseudorandomly in different
rounds, RICE produces a different $digest$ every round. ES nodes then
submits $(seed^{(j,\phi)},root(\sigma^*))$ as the $digest$ 
after executing $\Psi(\sigma,x)$, where $\phi$ is the total number of
times the $seed$ has been updated.
% \vinay{$L$ has been used for
%   likelihood in \miracle. Try to change if possible.}
%This ensures that the $digest$
%from different rounds can be correlated using their $root$ values.

Due to the deterministic nature of the $\Psi$, all nodes computing
$\Psi(\sigma,x)$ correctly will have the same $root$ across
rounds. The $seed$ values of all honest nodes will be identical within
any round, but will differ from one round to the next. Malicious nodes
may submit the correct $root$ but the wrong $digest$, an attack we
guard against in \textsection \ref{sub:reward}

%Note that it is possible for digests in the same round to have the $root$ value but different $seed$ values.
%where these indices are determined pseudorandomly based on $seed$ content. Hereon for better exposition we simply denote $\Psi(\sigma,\tau)$ with $\Psi$ unless explicitly required. 
%The intuition behind choosing this indices (pseudo)randomly is that for repeated $\Psi$ staring with two different $seed$ the chosen intermediate indices will vary and possibly so does the intermediate state at these indices. This leads to different post $\Psi$ $seed$ value for different rounds in MiRACLE. Choice on pseudo randomness of generating these indices is to make RICE uniform across all nodes in ES correctly following the protocol and doing so produce identical post $\Psi$ $seed$.

% {\bf VINAY: we have used $T$ for threshold earlier. We  use totalins elsewhere.}

\noindent{\bf Details.}
Let $t$ denote the indices in the instruction array where $t \in
[1:T]$ where $T$ is the total number of instructions executed as a
part of $\Psi(\sigma, x)$. 
Note that $T$ is unknown a priori, but assumed to be
bounded. Specifically, for systems such as Ethereum where CIC
transactions have a gas limit, $T$ is guaranteed to be bounded
(refer~\textsection \ref{sub:transaction agreement}). 
Thus to update $digest$,
instead of executing the entire $\Psi$ array in a single run, RICE
progressively executes a subarray of $\Psi$ array between two index
$t_i$ (initial) and $t_f$ (final), updates the $seed$, and repeats the
process with the next sub-array and so on until it reaches $T$.

Formally, let $\Psi[t_i:t_f]$ denote an arbitrary subarray from $\Psi$
with $t_i$ and $t_f$ its initial and final index. RICE consists of a
new deterministic contract execution function $\Psi'$ with the
following semantics.  Inputs to $\Psi'$ are two indices $t_i,t_f$, an
intermediate CIC state $\sigma'$ and $x$.  Given input $(t_i, t_f,
\sigma', x)$, $\Psi'$ executes subarray $\Psi[t_i:t_f]$ (both
$t_i,t_f$ inclusive) with state $\sigma'$ and data $x$. After
execution, $\Psi'$ returns a modified state and the last
successfully executed index.  In the special case where $T < t_f$ for
some $(t_i,t_f)$, $\Psi'(t_i, t_f, \sigma', x)$ runs only till
$\Psi[t_i:T]$ and returns $\sigma^*, T$ as its output. Formally,
\begin{equation*}
\begin{rcases}
  (\sigma', t_f) \text{, if } t_f < T\\
  (\sigma^*, T)\text{, if } t_f \ge T \\
\end{rcases}
\leftarrow \Psi'(t_i, t_f, \sigma', x) \text{, where }
\sigma^* = \Psi(\sigma,x)
\end{equation*}
After executing $(l+1)^{\rm th}$ subarray of $\Psi(\sigma, x)$, RICE
updates the seed via (\ref{eq:seedupdate}).
%assigns ${\sf hash}(seed^{(j, l)}||{\sf root}(\sigma'))$ to $seed^{(j,l+1)}$ with .
%The tuple ES nodes of round $j$ submits to the blockchain is $(seed^{(j,L)}, root(\sigma^*))$ where $L$ is the number of times $seed^{(j,.)}$ is updated during execution of $\Psi(\sigma, \tau)$.

Algorithm~\ref{algo:rice} gives the pseducode of RICE.
\begin{algorithm}
  \caption{RICE}\label{algo:rice}
  \begin{algorithmic}[1]
    \State \textbf{input} $seed^{(j,0)}, \sigma, x$
      \State $\sigma' \gets \sigma, l \gets 0$
      \State $t_i, t_f \gets $ \Call{NEXT}{} \Comment{Next subarray indices}
      \While{true}
      \State $\sigma', t_l \leftarrow \Psi'(\sigma', t_i, t_f, x)$
      \If{$t_l$ is $T$}{}
        \State {\bf return} $(seed^{(j,l)},root(\sigma'))$
      \EndIf
      \State $seed^{(j,l+1)} \gets {\sf hash}(seed^{(j,l)}||root(\sigma'))$
      \State $l \gets l+1$
      \State $t_i, t_f \gets$ \Call{NEXT}{}
      \EndWhile
  \end{algorithmic}
\end{algorithm}

\subsection{Choosing the indices.}
\label{sec:choosing indices}
A naive strategy is to choose indices $t_f$ as multiples of a fixed
number, say $\Delta$. Note that $\Delta$ cannot be a function of $T$
which is not known prior to computing $\Psi$. This strategy leads to
overheads of $O(T)$.

Another problem arises because indices do not change from one round 
to the next. Suppose a quasi-honest node of the current wants to 
free-load the $root$ values at these indices from an earlier round. 
It can ask any node from an earlier ES to provide these root values 
and use (\ref{eq:seedupdate}) to verify that they indeed 
corresponded to the $digest$ submitted by that node, thereby giving it 
confidence that these root values are correct. It can then reuse these 
root values to create its own digest without performing
the computation. In case the ES node queried is malicious, the 
quasi-honest node will submit an incorrect solution.

%As a result a node can {\it
%  free-load} by asking ES nodes from previous rounds to reveal the
%$root$ values at these indices to which an adversary $\mathcal{A}$
%might respond. Making things worse, $\mathcal{A}$ can provide proof
%that revealed roots are generated using the correct procedure of
%repeatedly hashing the $digest||root$ at different indices. Note that
%this procedure does not prove that the root values are correct and
%indeed correspond to the actual root values at the indices.  For an
%arbitrary round $j$, call all $root^{\mathcal{A}}_{k\Delta}$ the root
%values of $\mathcal{A}$ after the $k\Delta^{\rm th}$ instruction
%corresponding to $digest^{\mathcal{A}}$, all of which may be
%incorrect.  Then repeating $seed^{(j, k+1)} \leftarrow
%Hash(seed^{(j,k)}||root^{\mathcal{A}}_{k\Delta})$ starting with the
%initial known seed in its round and iterating over all $k$ it shows
%that $digest^{\mathcal{A}}$ corresponds to these root values. An ES
%node seeing this proof may be tempted to believe they are correct if
%the $digest(\mathcal{A})$ is the digest with highest {\it
%  log-likelihood} \textsection \ref{sec:miracle}. It may hence not
%execute $\Psi$ and instead simply reuse $root^{\mathcal{A}}_{k\Delta}$
%to compute its digest for current round. In case
%$digest^{\mathcal{A}}$ was not the digest corresponding to the correct
%execution of CIC, this may lead to an attack where an Byzantine
%adversary introduces a false digest by offering
%$root^{\mathcal{A}}_{k\Delta}$ values to ES nodes from an earlier
%round.

A second naive strategy is to choose the sub-array sizes randomly but
with mean size exponentially increasing as $\Psi$ progresses. For
example, choose $t_f - t_i$ randomly from $[2^k:2^{k+1}]$ where $k$
increments by 1 from one sub-array to the next. On the positive side,
this will lead to $O(log_2T)$ seed updates (and consequently overheads
of that order) and also will produce a different set of indices from
one round to the next with large probability.

% \vinay{gaslimit not yet introduced, so modifying below}
However, there remains the problem of skipping computing the last
sub-array of the instruction array.  Suppose a quasi-honest node in the current
round's ES has learned the value of $T$ from ES nodes of earlier
rounds. It can perform computation of $\Psi$ for all sub-arrays
except for the last one. Then it can use a value of $root$ submitted
in an earlier round in (\ref{eq:seedupdate}) to obtain the final
$seed$ value, without computing the last sub-array.
For this strategy the last sub-array can be as large as $T/2$,
leading to
nodes skipping as much as half of the computation. Hence although
overheads have reduced to $O(log_2T)$, the computation skipped at the
end is $O(T)$. We seek to find a sweet spot between the two with our
choice of indices for RICE.

%At any point near the end of $\Psi(\sigma, \tau)$, by looking at gas
%limit of $\tau$ ES nodes might be able to identify that the current
%$seed$ update is the last seed update and the next $seed$ update is
%beyond $\tau.gasLimit$. The $digest$ of a round requires $(seed,root)$
%but $seed$ after the last update remains fixed. By design the $root$
%already published on the chain with highest log-likelihood will be
%correct with large probability, and therefore nodes may skip computing
%beyond the last update by simply using $root$ with highest
%log-likelihood to generate their digests.  For this strategy the last
%$seed$ update can be at most $T/2$ prior to $T$ thereby leading to
%nodes skipping as much as half of the computation. Hence although
%overheads have reduced to $O(log_2T)$, the computation skipped at the
%end is $O(T)$. We seek to find a sweet spot between the two with our
%choice of indices for RICE.

\begin{algorithm}
  \caption{Next subarray indices}\label{algo:next indices}
  \begin{algorithmic}[1]
    \State {\bf Let} $K = [1,2,2,3,3,3,4,4,4,4,\ldots]$
    \Procedure{NEXT}{$seed, \sigma, index, t_f$}
    \State $k \gets K[index]$
    \State $pivot \gets t_f-{\sf int}(seed[1:k])+2^k$
    \State $nextSeed \gets {\sf hash}(seed||root(\sigma'))$ 
    \State $t_i \gets t_f+1$
    \State $t_f \gets pivot + {\sf int}(nextSeed[1:k])$
    \State {\bf return} $t_i, t_f$
    \EndProcedure
  \end{algorithmic}
\end{algorithm}

\noindent{\bf Our Approach.} RICE uses a hybrid of the two index
locating procedures described above. The idea is to divide the array
$\Psi'$ into sub-arrays of size $2^k$ where
$k=1,2,2,3,3,3,4,4,4,4,5,\ldots$. In other words, every value of $k$
repeats $k$ times.
%Consequently the sub-array sizes increase
%sub-exponentially.
%We choose one updating index $t_f$ in each
%sub-array and update $seed$ with {\it
%  Hash}$(seed||root(\sigma'.storage))$ where $\sigma'$ is the
%intermediate state of $\sigma$ after executing all instructions up to
%$t_f$.
Thus, like the second naive scheme the sub-array size increase,
but much more gradually (sub-exponentially) so that the last sub-array which might be
skipped is smaller.  More precisely, for a sub-array of size $2^k$ we
choose the index to update the $seed$ as ${\sf int}(seed[1:k])$ away
from the beginning of the sub-array, where ${\sf int}(seed[1:k])$
denotes the integer whose binary representation is identical to the
first $k$ bits of $seed$. Algorithm~\ref{algo:next indices} demonstrates
our approach.

% !TEX root = ../main.tex
\section{CIC Preliminaries}
\label{sec:cic prelims}
% In this section we describe preliminaries for
% off-chain execution of CICs.

\noindent{\bf Smart Contracts and its Execution.} 
% We model smart contracts as executable objects in the blockchain with
% the following semantics. 
A smart contract in YODA is denoted by its
state $\sigma=(cid, code, storage)$. Here $cid$ denotes its immutable
globally unique cryptographic identity, and $code$ represents its
immutable program logic consisting of functions. The state can be
modified by a transaction invoking its code and its execution can only
begin at a function. Functions may accept data from sources external
to the blockchain which must be included in the transactions invoking
them. In YODA smart contracts are stateful and state is maintained as
$(key,value)$ pairs which together we refer to as $storage$.

A transaction $\tau$ in YODA is the tuple $(tid, funId, data,
\xi)$. Here $tid$ is a globally unique transaction identity and
$funId$ is the function it invokes. All external inputs required
for the function are part of $data$ and $\xi$ consists of
meta-information about the account that generated the transactions
along with a cryptographic proof of its authenticity. Hereafter we
assume all transactions are validated using $\xi$ before being
included in a block and hence we drop $\xi$. 
% In YODA, transactions can
% be used to transfer tokens, create contracts or execute functions from
% smart contracts.

Executions of functions in YODA are modeled as {\it transaction driven
state transitions}. We use $\Psi$ to denote a {\it Deterministic
State Transition Machine } (Possibly Turing complete). Formally,
we denote this as $\sigma^* \leftarrow \Psi(\sigma, \tau)$ where
$\sigma^*$ is the state of the contract after executing $\Psi$.

\noindent{\bf Intensive Transactions.}  Intensive Transactions (IT)
are transactions which cannot be executed on-chain due to either of
two problems: its execution time exceeds the typical interspacing
between blocks, or it competes with PoW time (the Verifier's Dilemma).
The first problem can occur in permissioned ledgers such as
Hyperledger, Quorum, R3-Corda etc., and both problems in
permissionless blockchains such as Ethereum and Bitcoin. The exact
definition of an IT will depend on parameters of the blockchain system
under consideration. Transactions which are not ITs are called {\it
  non-ITs}.

We give one example of a IT for Ethereum using the concept of $gas$,
a measure of cost of program execution~\cite{buterin2014next}. 
Ethereum associates a fixed cost with each
machine level instruction that a smart contract executes and enforces
the constraint that all transactions included in a block can consume a
maximum combined gas of {\it blockGasLimit} which is set to prevent
the Verifier's Dilemma.  Every time a transaction $\tau$ is broadcast, 
its creator specifies $\tau.gasLimit$, an upper bound on the gas it is expected to
consume. Clearly $\tau.gasLimit<blockGasLimit$ for any transaction to
be included in a block. Transactions which violate this condition are
thus ITs. 
% \vinay{Next sentence is not clear. Do you mean non-ITs can be within a
 % CIC? Or non-ITs can be executed off chain? ``allows execution'' 
 % is not clear.} YODA allows execution of both ITs and non-ITs as long as
% they deposit a minimum fee described in~\textsection \ref{sub:cic
  % transaction}.

\noindent{\bf Computationally Intensive Contracts.}  We term all smart
contracts that execute ITs as Computationally Intensive
Contracts. Since YODA allows transactions of different CICs as well as
on-chain transactions to run in parallel we make some assumptions as
well as provide special mechanisms to prevent race conditions from
occurring. First, we assume that transactions of CIC that are run off-chain cannot
modify the storage of any other smart contract whether CIC or
non-CIC. Second, we assume that on-chain transactions cannot modify
the state of any CIC if any off-chain transaction execution of the CIC
is in progress. Third, we ensure that all transactions of all CICs are
ordered on-chain before their execution by a special on-chain smart
contract called the {\em Master Contract}\footnote{Systems can be
built where all rules in MC are part of the basic System protocol
instead of making it a smart contract. Our implementation makes MC a
smart contract.} (MC).

\noindent{\bf Master Contract.} The MC maintains a queue $Q_\sigma$ in
which all transactions (ITs and non-ITs) of CIC $\sigma$ are stored
before being executed.  The transaction at the queue's head is
executed first and a new incoming transactions added to its tail. In
case a non-IT is at the head of $Q_{\sigma}$, miners execute the
transaction on-chain otherwise the transaction is executed
off-chain.
%$Q_\sigma$ also demonstrate the CICs that are currently
%being executed off-chain. All CICs whose transaction queue head is an
%IT are assumed to be under off-chain execution and are is protected
%from on-chain modification.

The MC performs many tasks in YODA.
Since each CIC has its own queue, ITs of different CICs can be
executed in parallel off-chain. In addition, the MC embodies the rules
of YODA like creating CICs, ordering their transactions and initiating
off-chain execution, running \miracle, distributing rewards to
the ES node, enabling YODA nodes to join SP by collecting their
deposits etc.

\noindent{\bf Stake Pool.} YODA prevents sybil entries in
SP~\cite{douceur2002sybil}. To join SP, a node needs to deposit stake
$\mathcal{D}_{sp}$ by sending a transaction to MC. This deposit acts
as insurance for misbehavior of SP nodes. 
The SP once selected remains valid for a system defined
interval of time denoted by $\delta_{sp}$ beyond which YODA
re-initiates the SP selection procedure.  Note that SP could
potentially include all miners in the entire network, especially in a
small blockchain network maintained by several hundred nodes such as
blockchains that are built using
Hyperledger~\cite{androulaki2018hyperledger}.

% \vinay{Next sentence is not clear. We can remove if it is not required for the overall picture of YODA.} During re-balancing rewards to SP nodes that
% tries to withdraw from previous SP are only given upon termination of
% all CICs from the previous epoch. Nodes willing to continue to the
% next SP need to only send a transaction showing their willingness
% without making any further deposits.

\noindent{\bf Execution Set.}  YODA selects a subset of nodes from 
SP known as the Execution Set (ES) to execute $\Psi(\sigma, \tau)$.
%An ES is
%chosen from a Sybil resistant larger set called a {\it Stake Pool}
%(SP). Nodes join SP by depositing a fixed amount of tokens as their
%stake. YODA forfeits the deposit of misbehaving nodes using techniques
%developed in \textsection (CIC enabling).

\noindent{\bf CIC creation and deployment.}  To deploy a CIC with
state $\sigma$ on the blockchain, anyone can broadcast a transaction
requesting creation of CIC containing the tuple $(code,storage)$. 
Miners use ${\sf RandomGen()}$ to generate a
% seed^{(0,0)}$ and a 
unique identity $cid$ for the CIC. An entry is registered in the 
Master Contract (MC) corresponding to the new CIC along with an empty 
transaction queue $Q_\sigma$. Then miners deploy $\sigma$ on
the blockchain like any other smart-contract.
%With these
%preliminaries, in next section we describe the 5 steps for off-chain
%execution of CIC.

\noindent{\bf External Functions.} The following functions are used in 
YODA and are assumed to be accessible to all nodes. 
\begin{itemize}
  \item $(o,\pi)\leftarrow {\sf CheckSort}
  (pk,sk,nonce,threshold)$. This function on invocation internally runs
  Secret Cryptographic Sortition (SCS)~\cite{gilad2017algorand}. The
  $threshold$ is used to set the probability $q$ that an SP node is in
  ES. If ${\sf CheckSort()}$ returns $\perp$ it implies the node
  was not selected. Otherwise the node is selected and $o$ is
  indistinguishable from a truly random number to anyone without 
  $sk$~\cite{micali1999verifiable}. However, it is easy to prove that 
  ${\sf CheckSort()} \neq\perp$, given $\pi$ and $pk$.

  \item $r \leftarrow {\sf RandomGen()}$ on invocation produces an
  unbiased distributed random string. It can be practically built
  using the Randhound protocol given in \cite{syta2017scalable} or
  using block headers of a sufficiently long set of blocks. A simple 
  and efficient block nonce generation procedure is using $k$ 
  historical block hashes for sufficiently large $k$ as 
  in~\cite{gilad2017algorand,luu2016secure}. Alternatively, one can
  use a NIST randomness beacon relayed through a data feeding 
  mechanism such as Town~Crier~\cite{zhang2016town}.
  % \sourav{We primarily use ${\sf RandomGen()}$ to generate $nonce$ for Sortition on receiving an IT. Other usage of it
  % is during CIC deployment. Thus, we can generate the random numbers as follows. Generate a $nonce$ per block using
  % techniques mentioned in~\cite{algo}}
\end{itemize}

\section{Enabling CICs in Blockchain}
\label{sec:enabling cic}
In this section we describe how on receiving an IT $\tau$, YODA
executes it off-chain.  We have described two of YODA's key
ingredients in detail: \miracle, which enables efficient CIC
computation with small sets of nodes, and RICE which makes guessing
the seed of one round difficult from submitted digests in earlier
rounds. The other mechanisms we describe here address the other
challenges mentioned in~\textsection \ref{sub:challenges}, such as
preventing sybil attacks, collusion, DoS, and certain variants of
free-loading attacks.

\noindent{\bf S1. CIC Transaction Deployment.}
\label{sub:transaction agreement}
On receiving an IT, $\tau$, miners generate string $nonce$ using 
${\sf RandomGen()}$ which is used in ${\sf CheckSort()}$ to elect an
ES. It is important that the $nonce$ be created {\em during or only 
after} inclusion of $\tau$ on-chain. Otherwise, if the $nonce$ is
known a priori, the node generating $\tau$ can perform the following
attack to dominate the ES formed. It can enroll with key-pairs
$(pk,sk)$ in SP such that the Sortition Check~\textsection
\ref{sub:sortition check} results for the key-pairs will guarantee it
a large membership in ES. It then broadcasts the transaction,
dominates the resulting ES, and submits false solutions which may be
accepted.

The creator of $\tau$ deposits $(\mathcal{D}_{min}+gasPrice*gasLimit)$ 
in the MC where $\mathcal{D}_{min}$ is the minimum amount to pay for 
the fixed costs of S4 and S5, and $gasPrice$ and $gasLimit$ denote the 
gas price and gas limit respectively specified by $\tau$. Any extra 
deposit after execution of $\tau$ is refunded.

\noindent{\bf S2. Sortition Check and RICE of CICs.}
\label{sub:sortition check}
Nodes decide if they are in ES corresponding to $\tau$ by computing 
(\ref{eq:sortition check}) given below.  The node is part of the ES if and 
only if its $sort\_res\neq \perp$. All nodes selected for the ES then 
execute the corresponding CIC in RICE as given in~\textsection 
\ref{sec:rice}, generate the corresponding RICE digest and proceed to 
S3.
\begin{equation}
sort\_res \leftarrow {\sf CheckSort}((pk,sk),\tau.nonce, threshold)
\label{eq:sortition check}
\end{equation}

By using SCS, YODA prevents Byzantine nodes from joining an ES at will.
% Thus SCS along with \miracle\ prevents the issue of {\em Increased 
% Malicious Fraction}. \vinay{Has this Increased Malicous fraction issue been discussed earlier? If
% not remove the sentence.} 
Also SCS protects ES nodes from DoS attacks
since their selection is secret until they reveal the fact. Because
YODA uses a commit-reveal mechanism (see below), ES nodes are not
vulnerable to DoS attacks before they submit their commit
transaction. After the commit step, ES nodes may be easier to identify
and hence the DoS attack can be more effective. In the second 
``reveal'' step, an ES node broadcasts a single on-chain transaction 
after a certain number of blocks have been generated.  We assume that 
a node is sufficiently DoS resilient to be able to receive block 
headers in a timely manner and also to broadcast a small transaction.

% \vinay{Have we described these challenges anywhere?}
% Since any ES is small, we face challenges like {\it Increased
%   Malicious Fraction}, and {\it Low cost DoS attacks} described in
% \textsection \ref{sub:challenges}. 
% in ES is partially mitigated using Sortition which selects nodes at
% random. This prevents Byzantine nodes from joining the ES at will.  It
% is further addressed by \miracle\ which can tolerate occasional large
% fractions of malicious nodes in ES (see \textsection \ref{sec:miracle}).

\noindent{\bf S3. Commitment and Release.}
\label{sub:commitment}
Let $digest_k$, and $sort\_res_k$ denote the digest of
$\Psi(\sigma,\tau)$, and the result of ${\sf CheckSort()}$ respectively
for node $n_k \in ES$ for $k=1,2,\ldots,|ES|$. The commitment $n_k$
generates is $se_k$ and is given by
\[se_k \leftarrow {\sf hash}(digest_k||sort\_res_k)\]
Assuming existence of a VRF and Ideal Hashing, $w.h.p$ 
$sort\_res_k \neq sort\_res_i$ for $i\neq k$ and hence 
$se_i \neq se_k $ even if $digest_k = digest_i$. Nodes in ES then 
broadcast $se_k$ to the blockchain as a transaction which miners 
include on-chain if the node is in SP.

ES nodes must broadcast their commitment within a time window $w_{src}$
or commitment period starting from the block that includes
$\tau$. Window $w_{src}$ is transaction dependent, recorded in MC and
measured in number of blocks.  It is set based on the gas limit
mentioned in $\tau$ as by design $\Psi(\sigma, \tau)$ will run for at
most $\tau.gasLimit$ instructions.

A node is required to keep $sort\_res_k$ secret during this period and 
forfeits its deposit if it fails to do so. This deters ES nodes from 
colluding. However, as mentioned in \textsection \ref{sub:threat 
assumption}, ES nodes can use of ZK-proofs to prove that they are in an 
ES without revealing $sort\_res_k$. We perform a game theoretic 
analysis of such an attack in \textsection \ref{sub:collusion}.

After $w_{src}$, nodes in ES wait for a buffer period $w_{buf}$ before
sending their unhashed digests and sortition results to the
blockchain. Nodes in ES which have submitted commitments earlier, are
required to submit a transaction containing $(digest_k||sort\_res_k)$
within a time window $w_{sr}$ or release period and failure to do so
results in their forfeiting deposits and being removed from SP.

% \vinay{This is a new attack, not part of the Threat model. 
% We had assumed that the blockchain has certain immutability 
% properties. We can put a note if the CFA violates our Threat model.}
% \sourav{CFA do violate our threat model. I think we can just one line
% stating, CFA will not occur for blockchains with immediate confirmation
% guarantees. Thus for such blockchains $w_{buf}=0$.}
% \sourav{Add footnote}
The reason for keeping this buffer period of length $w_{buf}$ is to
prevent an adversary from launching a DoS attack which we term the
{\em Chain Forking Attack} (CFA).\footnote{According to the blockchain
  safety and
  liveness assumptions of our threat
model, CFA will never arise and we 
can safely consider $w_{buf}=0$. However for blockchains reminiscent 
to Ethereum, CFA is a practical concern, thus we have added a non-zero 
$w_{buf}$ in our implementation~\textsection\ref{sec:evaluation}.} 
CFA can occur in blockchains where block creation does not guarantee 
block finalization and nodes need to wait for certain number of blocks 
before becoming certain about a block's finality $w.h.p$. Assume the 
absence of this buffer period. If this buffer did not exist then if an 
honest node publishes its opened commitment after $w_{src}$ and 
expects its inclusion in a future block, an adversary can create an 
alternate chain where it includes this transaction before the end of 
$w_{src}$ and can thereby penalizes the honest node. To prevent this 
from happening, the introduction of $w_{buf}$ between $w_{src}$ and 
$w_{sr}$ ensures that the attacker will have to create a fork which is 
long enough to be prohibitively expensive to create.

\noindent{\bf S4. \miracle\ for CICs.}
\label{sub:miracle for cic}
The blockchain miners then execute one round of \miracle\ using the 
submitted digests. All digests with the same $root$ are considered by 
\miracle\ to be the same solution irrespective of which
$seed$ they contain. Steps S2-S4 are repeated if necessary till 
\miracle\ converges. 

\noindent{\bf S5. State Update, Reward Distribution and Cleanup.}
\label{sub:reward}
% \noindent{\bf State Update.} 
Once \miracle\ terminates, all nodes in the ES broadcast one or more 
transactions to the blockchain miners containing information required for 
updating the state of the CIC to $\sigma^*$, the state corresponding 
to the winning digest along with corresponding proof.\footnote{For a 
Merkle tree implementation of $\sigma$ it will suffice to send only 
the Merkle paths of all modified variables in $\sigma$.} 
Any miner, on receiving such a transaction validates it using the root 
contained in the winning $digest$ in \miracle\ and then gossips it to
other miners. To avoid flooding the network, it does not gossip any
other transactions
about the same state.

%On first successful 
%validation miners updates $\sigma$ blockchain with $\sigma^*$. Further, they 
%gossip the corresponding transaction to their neighbors identical to 
%any other transactions. To avoid flooding of network, with
%identical state update transaction, miners only gossip the first
%correct state to their neighbors.

% \noindent{\bf Reward Distribution.}
To disincentivize nodes from submitting either false $seed$ or $root$
values, YODA rewards ES nodes as follows.  Let $I$ denote the round
in which \miracle\ terminates with $root$.  The deposits of all ES
nodes who submitted a digest with different root from the winning one
are forfeited.  For a round $i|i\le I$, let $seed_k| k \in
\{1,2,\ldots K\}$ be the $K$ different $seed$ values submitted in
digests containing $root$ and let $e_k$ be their count. YODA then
rewards only the ES nodes corresponding to the $seed_i$ for which
$\frac{e_i}{\sum_{j=1}^{K}e_j} > th_1$ where $th_1 > 0.50$. YODA
confiscates the deposit of all ES nodes for which
$\frac{e_i}{\sum_{j=1}^{K}e_j} < th_2$ and YODA neither rewards nor
punishes the rest. These forfeited deposits are either burned or
transfered to the MC.

The intuition behind using thresholds is as follows. Although \miracle\
identifies the correct root with probability $1-\beta$, it cannot say which of many digests
(with different $seed$ values), both containing the correct root in a
particular round, is correct. Rewarding both would encourage
free-loading. A naive solution would be to reward the set of nodes
corresponding to $\max_k \{e_k\}$ and punish the rest. There are, 
however, rare instances in which Byzantine nodes can exceed 
quasi-honest nodes in a round. 
If Byzantine nodes publish the correct $root$ but with an
incorrect seed, the naive method would severely punish the honest
nodes. The set of quasi-honest nodes will however not be a very small
fraction of an ES. Hence threshold $th_2$ is chosen small enough to
ensure that quasi-honest nodes which behave honestly will not be
punished $w.h.p$, while at the same time punishing lone quasi-nodes
which try to guess the correct seed. Quasi-nodes have to collude in
large numbers to cross the $th_1$ threshold, an attack which is
non-trivial and analyzed in \textsection \ref{sub:collusion}

% \noindent{\bf Cleanup.}  
Lastly, blockchain miners perform a cleanup, where they deallocate
space used for execution of $\Psi(\sigma,\tau)$. This includes the
space for storing commitments, sortition results etc. Following this
miners check whether the transaction queue $Q_\sigma$ is empty or
not. If it is non-empty, SP nodes to initiates the protocol for 
off-chain execution of the transaction at the head of the queue and 
the cycle continues. 

\section{Security Analysis}
\label{sec:analysis}
In this section we analyze the security properties of YODA. We first
analyze \miracle\ and prove many results, the most important being that
it is optimal in the expected number of rounds under certain constraints. 
Incidentally, given \miracle, Byzantine nodes maximize the probability 
of choosing an incorrect solution by all submitting the same incorrect 
solution. Ironically, \miracle\ is optimal given this particular 
strategy of Byzantine nodes.

We then analyze RICE, proving bounds on the number of update indices,
and the amount of computation that can be skipped at the end. We also
prove that w.h.p. every round will have update indices which have not
been encountered in previous rounds. This makes free-loading
difficult.

We then present a Game-theoretic analysis of our incentive schemes
proving them to have Nash Equilibria \cite{nash1950equilibrium}. We
finally stitch together all our results to show how they meet the
goals mentioned in \textsection \ref{sec:introduction}. Lastly
we discuss why guarantees in YODA are likely to work in an even more
realistic setting with a stronger adversary and where quasi-honest 
nodes are  allowed more protocol deviations.

\subsection{\miracle\ Analysis}
\label{sub:miracle analysis}
In this section we present the security analysis and guarantees 
provided by \miracle. Let $M$ be the total size of SP, i.e. $M=|SP|$ 
containing $f$ fraction of Byzantine nodes. Let the probability of any 
node in SP getting chosen for an ES be $q$. Let $N^b, N^h$ denote the 
total number of Byzantine and quasi-honest (here assumed to be honest) 
nodes in an ES. Let  $N^b_i, N^h_i$ denote $i.i.d$ Bernoulli random 
variable indicating if the $i^{\rm th}$ Byzantine or honest node is 
selected for the ES or not. Then
%Let set $H \subseteq SP$ denote the set of NBR nodes in SP and correspondingly $B \subseteq SP$ denote the byzantine nodes. Then we denote
%\[ H = \{h_1, h_2,...., h_{(1-f)M}\} \]
%\[ B = \{b_1, b_2,...., b_{fM}\} \]
%Then $N^b, N^h$ denote the number of byzantine and NBR nodes in ES. Let $N^b_i, N^h_i$ denote $i.i.d$ Bernoulli random variable indicating if $h_i$ is $b_i$ is selected in ES or not, then
\[ N^b = \sum_{i=1}^{fM} N^b_i \text{, and }
N^h = \sum_{i=1}^{(1-f)M} N^h_i \]
We approximate $N^h, N^b$ by a Gaussian distribution, since they 
are a sum of large number of $i.i.d$ random variables. Let 
$\mu_h, \mu_b$ denote the mean of $N^h,N^b$ respectively and $\nu^2_b,\nu^2_h$ 
denote their variances. These are
$\mu_h = E[N^h] = q(1-f)M$, and  $\nu^2_h = Var[N^h] = q(1-f)M(1-q)$.
Similarly, $\mu_b = E[N^b] = qfM$, and $\nu^2_b = Var[N^b] = qfM(1-q)$.

\begin{theorem}
If $f=f_{max}$ and the adversary submits only a single incorrect 
digest, then \miracle\ reduces to an optimal Sequential Probability 
Ratio Test (SPRT)~\cite{wald1973sequential}.
%Let $\beta$ denote the allowed probability of accepting a incorrect storage root
%which is a design parameter. 
%Recall that the threshold $T$ for \miracle\ is
%\begin{equation}
%\label{equ:T}
%    T = \left(\ln\frac{1-\beta}{\beta}\right)\frac{2q(1-q)M(1-f_{max})f_{max}}{(1-f_{max})-f_{max}}.
%\end{equation}
The threshold (see (\ref{equ:T})) provides an optimal expected number of rounds for a given $\beta$. The expected number of rounds is given by
\begin{equation}
    \mathbb{E}[\text{\# of rounds}] \approx \frac{(1-\beta)\ln\frac{1-\beta}{\beta}+\beta\ln\frac{\beta}{1-\beta}}{\frac{(\mu_h - \mu_b)^2 + \nu^2_h - \nu^2_b}{2\nu^2_b} + \ln\frac{\nu_b}{\nu_h}}
    \label{eq:expectedround}
\end{equation}
\end{theorem}
\begin{proof}
Consider the case where all Byzantine nodes consistently provide the same solution. Let the solutions be $d_1$ and $d_2$. The problem of determining the correct solution boils down to choosing between two hypotheses over multiple rounds:
$\mathcal{H}_k: d_k$ is the correct solution; $k=1,2$. Let $c_{k,j}$ denote the number of solutions equal to $d_k$ in round $j$.
Then the optimal solution \cite{wald1973sequential} is given by an SPRT in which  the log-likelihood ratio after $i$ rounds is
\[L_i = \sum^i_{j=1} -\frac{(c_{1,j} - \mu_h)^2}{2\nu^2_h}- \frac{(c_{2,j} - \mu_b)^2}{2\nu^2_b} + \]\[\frac{(c_{1,j} - \mu_b)^2}{2\nu^2_b} + \frac{(c_{2,j} - \mu_h)^2}{2\nu^2_h} \]
%\[
%= \sum^i_{j=1} \frac{c^2_{2,j} - c^2{1,j}}{2}\left[ \frac{1}{\nu^2_h} - \frac{1}{\nu^2_b}\right] - (c_{2,j}-n_{1,j})\left[ \frac{\mu_h}{\nu^2_h} - \frac{\mu_b}{\nu^2_b}\right]
%\]
\[
= \frac{1}{2q(1-q)M}\left[\frac{(1-f_{max})-f_{max}}{(1-f_{max})f_{max}} \right] \sum^i_{j=1}(c^2_{1,j}-c^2_{2,j})
\]
If $L_i > \ln((1-\beta)/\beta)$, 
then the SPRT chooses $\mathcal{H}_i$. This is equivalent to \miracle.
%Notice that this is the equation we have used in \miracle\ whose optimal solution is a SPRT. Hence \miracle\ reduces to an optimum SPRT guarantee with a guarantee that $P(\text{incorrect solution}) \leq \beta$
%Thus we have
%\[ L_i > \ln\frac{1-\beta}{\beta}\]
%\[ = \sum^i_{j=1}(c^2_{1,j}-c^2_{2,j}) \geq \ln\left(\frac{1-\beta}{\beta}\right){2q(1-q)M}\left[\frac{(1-f)f}{(1-f)-f} \right] \]
\end{proof}

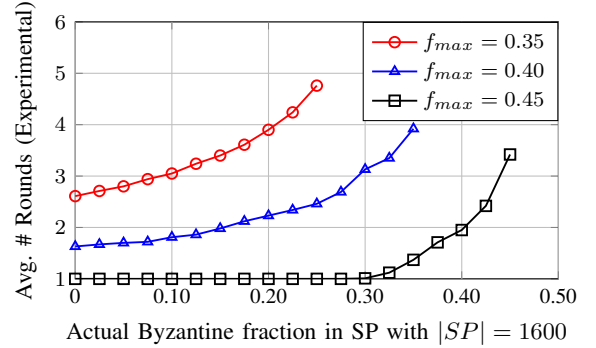
\begin{figure}[h!]
    \centering
    \pgfplotsset{footnotesize,height=5cm, width=0.9\linewidth}
    \begin{tikzpicture}
    \begin{axis}[
    legend pos=north west,
    xmin=0.0,
    xmax=0.5,
    ymin=1,
    ymax=6,
    xlabel=Actual Byzantine fraction in SP with {$|SP|=1600$},
    ylabel= Avg. \# Rounds (Experimental),
    grid=major,
    legend style={at={(1,1)},anchor=north east}, 
    xtick = \empty,
    extra x ticks={0,0.10,...,0.55},
    extra x tick labels = {0,0.10,0.20,0.30, 0.40, 0.50},
    ]    
    \addplot [line width=0.25mm, red, mark=o] coordinates {(0.0, 2.61) (0.025, 2.71) (0.050, 2.80) (0.075, 2.94) (0.10, 3.05) (0.125, 3.24) (0.150,3.40) (0.175,3.61) (0.20, 3.90) (0.225, 4.24) (0.250, 4.76)};

    \addplot [line width=0.25mm, blue, mark=triangle] coordinates {(0.0, 1.63) (0.025, 1.67) (0.050, 1.7) (0.075, 1.72) (0.10, 1.81 ) (0.125, 1.86) (0.150, 1.98) (0.175,2.12)
    (0.20,2.23) (0.225,2.34) (0.250,2.46) (0.275,2.69) (0.30,3.13) (0.325, 3.35) (0.350, 3.92)};

    \addplot [line width=0.25mm, black, mark=square] coordinates {(0.0, 1) (0.025,1) (0.050,1) (0.075,1) (0.10,1 ) (0.125,1) (0.150,1) (0.175,1) (0.20,1) (0.225,1) (0.250,1) (0.275,1) (0.30,1.01) (0.325,1.12) (0.350,1.37) (0.375,1.71) (0.40,1.95) (0.425, 2.42) (0.450,3.42)};
    \addlegendentry{$f_{max}=0.35$}
    \addlegendentry{$f_{max}=0.40$}
    \addlegendentry{$f_{max}=0.45$}
    \end{axis}
    \end{tikzpicture}
    \caption{Number of rounds \miracle\ takes to terminate with $f\le f_{max}$ when
    designed for worst case $(f=f_{max}$) expected number of rounds for termination 
    to be $5$ i.e $R(f_{max})=5$. For the case $f_{max}=0.45$, for all $f\le 0.30$,
    \miracle\ terminates in one round.}
    \label{fig:miracle adaptive}
\end{figure}

\noindent{\bf Remark 1:} \miracle\ is optimal in case $f=f_{max}$. 
In case $f<f_{max}$, the expected number of rounds will be less than 
that specified in~(~\ref{eq:expectedround}), while still ensuring that 
the probability of incorrectly deciding $\sigma^*$ is less than $\beta$.
Specifically, let $R(f)$ be the expected number of rounds \miracle\ takes 
to terminate as a function of $f$. In Figure~\ref{fig:miracle adaptive} ,
for $\beta=10^{-20}$,
we choose $q$ such that, $R(f_{max})=5$. We then empirically
evaluate the number of rounds \miracle\ takes to terminate when 
$f \le f_{max}$ over 10000 runs of \miracle. Key points to observe in the 
Figure~\ref{fig:miracle adaptive} is that \miracle\ terminates early
for $f \le f_{max}$. Specifically, with $f_{max}=0.45$, on average
\miracle\ terminates in one round up to $f\approx 0.30$.

\noindent{\bf Remark 2:} The probability of choosing a node to belong
to the ES, $q$, can be set to any value which fixes the expected size
of ES in any round. We recommend that $q$ be chosen such that if all
nodes in the first ES are honest then the likelihood crosses the
threshold $T$ in that round itself. This ensures that more than one
round can occur only if there are some Byzantine nodes in SP. Solid
line from graph in Figure~\ref{fig:miracle qpickup} demonstrates
minimum required $|ES|$ with $|SP|$=1600 and $f=0$ for one round
termination of \miracle.

\begin{figure}[h!]
    \centering
    \pgfplotsset{footnotesize,height=5cm, width=0.9\linewidth}
    \begin{tikzpicture}
    \begin{axis}[
    legend pos=north west,
    xmin=0.0,
    xmax=0.45,
    ymin=0,
    ymax=1200,
    xlabel={$f_{max}, |SP|=1600$},
    ylabel= {$\mathbb{E}[|ES|]$},
    grid=major,
    xtick = \empty,
    extra x ticks={0,0.10,...,0.55},
    extra x tick labels = {0,0.10,0.20,0.30, 0.40, 0.50},
    ]    
    % \addplot [line width=0.25mm, red, mark=o] coordinates {( 0.05 , 2.43 )( 0.10 , 5.16 )( 0.15 , 8.34 )( 0.20 , 12.19 )( 0.25 , 17.08 )( 0.30 , 23.82 )(0.35 , 34.18)(0.40 , 53.42)(0.45 , 106.40)};
    
    % \addplot [line width=0.25mm, dashed,  mark options=solid, red, mark=o] coordinates {( 0.05 , 2.70 )( 0.10 , 6.45 )( 0.15 , 11.89 )( 0.20 , 20.21 )( 0.25 , 33.81 )( 0.30 , 58.24 )( 0.35 , 108.51 )( 0.40 , 235.62 )( 0.45 , 665.62 )};

    % \addplot [line width=0.25mm, blue, mark=triangle] coordinates {( 0.05 , 3.64 )( 0.10 , 7.73 )( 0.15 , 12.48 )( 0.20 , 18.21 )( 0.25 , 25.49 )( 0.30 , 35.46 )( 0.35 , 50.72 )( 0.40 , 78.81 )( 0.45 , 154.46 )};

    % \addplot [line width=0.25mm, dashed, mark options=solid, blue, mark=triangle] coordinates {( 0.05 , 4.04 )( 0.10 , 9.66 )( 0.15 , 17.77 )( 0.20 , 30.12 )( 0.25 , 50.18 )( 0.30 , 85.80 )( 0.35 , 157.43 )( 0.40 , 329.19 )( 0.45 , 826.51 )};

    % \addplot[line width=0.75, dotted, red, mark options=solid, mark=o] coordinates {(0.05,35) (0.10,57) (0.15,89) (0.20,135) (0.25,211) (0.30,353) (0.35,655) (0.40,1525) (0.45, 1600)};

    \addplot [line width=0.25mm, black, mark=square] coordinates {( 0.05 , 4.85 )( 0.10 , 10.29 )( 0.15 , 16.60 )( 0.20 , 24.19 )( 0.25 , 33.81 )( 0.30 , 46.94 )( 0.35 , 66.92 )( 0.40 , 103.38 )( 0.45 , 199.53 )};

    \addplot [line width=0.25mm,dashed,  mark options=solid, black, mark=square] coordinates {( 0.05 , 5.38 )( 0.10 , 12.85 )( 0.15 , 23.61 )( 0.20 , 39.91 )( 0.25 , 66.22 )( 0.30 , 112.39 )( 0.35 , 203.24 )( 0.40 , 410.75 )( 0.45 , 940.13 )};

    % \addplot [line width=0.25, dotted] coordinates {(0.05,21) (0.10,35) (0.15,55) (0.20,85) (0.25,133) (0.30,223) (0.35,417) (0.40,971) };

    \addplot[line width=0.75, dotted, red, mark options=solid, mark=*] coordinates {((0.05,49) (0.10,81) (0.15,123) (0.20,187) (0.25,291) (0.30,483) (0.35,897) (0.40, 1600) (0.45, 1600)};

    \begin{tiny}
    % \addlegendentry{$\beta=10^{-10}, f=0$}
    % \addlegendentry{$\beta=10^{-10}, f=f_{max}$}
    % \addlegendentry{$\beta=10^{-15}, f=0$}
    % \addlegendentry{$\beta=10^{-15}, f=f_{max}$}
    % \addlegendentry{$\beta=10^{-15}$, NS1}
    \addlegendentry{$f=0$}
    \addlegendentry{$f=f_{max}$}
    \addlegendentry{NS1}
    \end{tiny}
    \end{axis}
    \end{tikzpicture}
    \caption{Dashed line shows $\mathbb{E}[|ES|]$ required for $R(f_{max})=1$. 
        Solid line shows $\mathbb{E}[|ES|]$ when $R(0)=1$. For the case 
        $f_{max}=0.35$, $|ES|\approx 200$ guarantees $R(f_{max})=1$. In addition 
        with $\mathbb{E}[|ES|] \approx 60$, for the same $f_{max}$, \miracle\ 
        terminates in $1$ round when $f=0$ and $\beta=10^{-20}$.}
    \label{fig:miracle qpickup}
\end{figure}
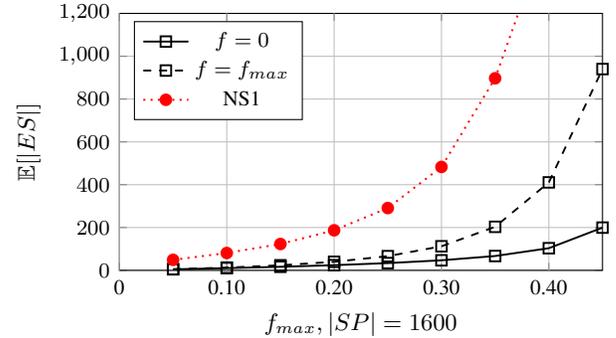

Although we have proved that \miracle\ is optimal if an adversary chooses a single solution, the question arises as to whether the adversary has a better strategy for \miracle\ in which it chooses more than one solution. The next theorem states that this is not the case. Indeed, the best strategy for the adversary, given that YODA uses \miracle\ is to choose only a single incorrect solution.

\begin{theorem}
With \miracle\ as the consensus algorithm, the best strategy for an adversary controlling all byzantine nodes is to only submit a single incorrect solution. Any other strategy reduces the probability of choosing incorrect solution by the system.
\end{theorem}
\begin{proof}
Consider a strategy in which adversary submits solution $d_2,d_1,...,d_m$ for $m\geq2$ and the solution submitted by honest nodes is $d_1$.
Let us call this strategy  ST1. Consider another strategy ST2, in which the adversary only submits a single solution denoted by $s'_2$ and the honest nodes submits $s'_1$.
For an round $i$ in ST2, let the corresponding likelihood be $L'_{2,i}$ and the number of solutions submitted is $c'_{2,i}$.
We assume that total number of submissions for an round $j\leq i$ are the same in both ST1 and ST2.
Hence, it is trivial to see that $c_{k,j} \leq c'_{2,j} \forall \ k \geq 2 $. Also,
\[ L_{k,i} = \sum^i_{j=1}(2c_{k,j}-C_j)C_j
\leq \sum^{i}_{j=1}(2c'_{k,j} - C_j)C_j\]
\[= \sum^{i}_{j=1}(2c'_{k,j} - C'_j)C'_j \]
\[ = L'_{2,i}, \forall\ k \geq 2, \forall i\]
Thus
\begin{equation}
    \max_{k\geq2} L_{k,i} \leq L'_{2,i}
\end{equation}
Hence the result follows.
\end{proof}

For each solution submitted \miracle\ has one likelihood which is compared with $\mathbb{T}$. The question arises as to whether or not more than one likelihood can simultaneously exceed the threshold, thus leading to multiple solutions for the transaction. The following theorem proves that this cannot happen.
\begin{theorem}
\miracle\ can terminate with only a single solution being adjudged correct.
\end{theorem}
\begin{proof}
Assume there are multiple solutions $d_1,d_2,..,d_t$ with $t\geq2$ for which $\sum_j{(2c_{k,j}-C_j)C_j} > \mathbb{T} > 0$. Then summing them
    \[ \sum^{t}_{k=1} \sum_j{(2c_{k,j}-C_j)C_j} > 0 \]
    \begin{equation}
        \Rightarrow \sum_k\sum_j {2c_{k,j}C_j} > t\sum_j{C^2_j}.
        \label{eq:zpac1}
    \end{equation}
    Consider the left hand side of the previous equation.
    \begin{equation}
        2\sum_jC_j\sum_kc_{k,j} \leq 2\sum_jC_j(C_j) \leq t\sum_jC^2_j
        \label{eq:zpac2}
    \end{equation}
    Equation \ref{eq:zpac1} and \ref{eq:zpac2}  contradict each other. Hence our assumption was wrong.
\end{proof}

\subsection{RICE Analysis}
\label{sub:rice analysis}
In this section we prove that RICE adds low overhead and is secure. 

\begin{lemma}
\label{lemma:rice1}
% {\bf $l$ may be used elsewhere}
Given RICE terminates in an subarray of size $2^k$, let $\phi$ be the number of times $\Psi(\sigma,\tau)$ is interrupted to update the $seed$ in RICE, then 
\begin{equation}
\label{eq:rice1}
\frac{(k-1)k}{2} < \phi \le \frac{k(k+1)}{2}
\end{equation}
\end{lemma}
\begin{proofsketch}
Due to the slow $k$ increase strategy, the total number of times storage root is updated i.e $\phi$ is 
    \[ \sum_{j=1}^{k-1}j < \phi \leq \sum_{j=1}^{k}j \]
  which proves the lemma.
\end{proofsketch}

% {\bf VINAY: Maybe push some of this to the Appendix. Not clear how this sub-section fits in with the rest}

% \begin{theorem}\label{aplem:ricel} For RICE,% with $\zeta^l$ being the highest value of rice where $l$ denotes the number of times storage root is updated then 
% \begin{equation}
% \frac{(\zeta^l-1)\zeta^l}{2} < l \leq \frac{\zeta^l(\zeta^l+1)}{2}
% \label{equ:ricel}
% \end{equation}
% \end{theorem}
% \begin{proof}
% Due to the slow $k$ increase strategy, the total number of times storage root is updated i.e $l$ is 
%     \[ \sum_{j=1}^{\zeta^l-1}j < l \leq \sum_{j=1}^{\zeta^l}j \]
%   which proves the lemma.
% %  Solving we get, 
% %    \begin{equation} 
% %    \frac{(\zeta^l-1)\zeta^l}{2} < l \leq \frac{\zeta^l(\zeta^l+1)}{2}.
% %    \end{equation}
% \end{proof}

Given the bounds on number of times $seed$ is updated in RICE, we now proceed to find a relationship between the $T$ and $\phi$. Thus we first find relationship between $T$ and $k$ and then proceed further. 
%where $total\_ins$ denote the serial length of $\Psi(\sigma,\tau)$ and $\zeta^l$ is described as above.

\begin{lemma}
\label{lemma:rice2}
The relationship between $k$ and $T$ is
\begin{equation} 
2^{k}(k - 2)+2 < T \leq 2^{k+1}(k-1)+2
\end{equation}
\label{aplemma:maxkrelation}
\end{lemma}
\begin{proofsketch} From the slow $k$ increase strategy we have
\begin{equation} \sum_{j=1}^{k-1}{j2^j} < T \leq \sum_{j=1}^{k}j2^j
\label{eq:ric1}
\end{equation}
Simplifying further we have the result.
\end{proofsketch}

\begin{theorem} {\bf (RICE Efficiency)} The number of times $seed$ is updated in a single RICE run i.e $\phi$ is $\Theta((\log_2T)^2)$.
\end{theorem}
\begin{proofsketch}
Taking log on both sides of the result of Lemma \ref{aplemma:maxkrelation} we have  
  % \[ 2^{\zeta^l}(\zeta^l-2) + 2 < total\_ins \leq 2^{\zeta^l+2}(\zeta^l-1)+4
%   \]
%    Taking $\log_2$ on both side, we get
  %  \[ \zeta^l + \log_2(\zeta^l-2) < \log_2(total\_ins) \leq \zeta^l+2+ \log_2(\zeta^l -1)+2 \]
    %Dropping lower order terms from above we arrive at 
$k$ is  $\Theta(\log_2T)$.
    This combined with equation \ref{eq:rice1} from Lemma \ref{lemma:rice1} proves the theorem.
 %  \begin{multline}
 %      \frac{(\log_2(total\_ins)-1)\log_2(total\_ins)}{2} < l\\
 %       \leq \frac{\log_2(total\_ins)(\log_2(total\_ins)+1)}{2} 
 %       \label{eq:tmaxrelationship}
 %  \end{multline}
 %  Hence from equation \ref{eq:tmaxrelationship}, we get 
 %  \begin{equation}
 %   l = \Omega((\log_2(total\_ins))^2)
 %  \end{equation}    
\end{proofsketch}

We now identify the possible point where last $seed$ update happens in RICE for $\Psi(\sigma,\tau)$

\begin{lemma}
With $T$ as the length of array representing $\Psi(\sigma,\tau)$, the last $seed$ update happens at $O(\frac{1}{log_2T})$ fraction prior to $T$. 
\end{lemma}
\begin{proof}
Let $t_l$ be the index of last seed update and let $t_l$ lies inside a segment of length $2^k$ then $T \le t_l + 2^k+2^{k+1}$. Hence we want $\frac{T-t_l}{T}$ is $O\left(\frac{1}{log_2T}\right)$
\[ T > \sum_{i=1}^{k}i2^i = k2^{k+2}+2 \]
\[ \frac{T-t_l}{T} \le \frac{2^k+2^{k+1}}{T} < \frac{3}{4k} \le \frac{3}{4log_2T} \]
\end{proof}

As discussed earlier, it is important for different rounds to use a
different set of indices for updating the seed to prevent free-loading
attacks. We thus prove how RICE prevents free-loading attacks in YODA.

% {\bf T and totalins}

\begin{definition} {\bf (Unmatched index)}
Let RICE$_i$ and RICE$_j$, with $i\ne j$ be two RICE runs in distinct rounds in YODA. 
%Recall that $T$ denotes the length of array representing $\Psi(\sigma,\tau)$. 
%Let $\phi$ denote the number of times $seed$ is updated in RICE.
Then $M_i = \{ t^m_i |\ m \in [1:\phi]\}$ and $M_j = \{t^m_j|\ k \in [1:\phi] \}$ denote the set of indices where $seed$ is updated in rounds $i$ and $j$.
% Denote the corresponding values of $r$ after being updated at these indices as $r^k_i$ and $r^k_j$.
An index $t^m_i \in M_i$ is said to be an {\it Unmatched index} with respect to RICE$_j$ iff $t^m_i \not\in M_j$.
% Then RICE$_i$, RICE$_j$ are said to have {\em unmatched indices} if $T_i \not \subseteq T_j$. If this is not the case for a pair of RICE$_i$,RICE$_j$ we say that they have {\em matching indices}. 
\end{definition}
% {\bf VINAY: $l_i$ was used in one place in Miracle. We can remove that, I think it was in an Algorithm. There it is not clear what $l_i$ is.}

% {\bf NEXT THEOREM IS NOT REQUIRED HENCE COMMENTING OUT}
%\begin{theorem}
%GIVE RICE$_i$ and RICE$_j$ with $i\ne j,$ $w.h.p\ \exists$ an Unmatched index in RICE$_i$ w.r.t RICE$_j$ and vice versa.
%\[ P\left[\text{RICE$_i$, RICE$_j$ $(i > j)$ have unmatched index}\right]
%\geq 1-\frac{1}{2^\lambda}  \] 
%\[ \text{where } \lambda = \frac{(k-1)k(2k - 1)}{6} \]
%
%\end{theorem}
%\begin{proof} Assuming ideal Hashing, for each sub-array of size $2^u$ in $\Psi{\sigma, \tau}$, the indices of both RICE$_i$ and %RICE$_j$ are identical with probability $\frac{1}{2^j}$. Let $A^u_v$ denote the event that during $v^{th}$ subarray of size $2^u$ the %indices both RICE$_i$ and RICE$_j$ are same. Let $A$ be the event that the indices matches in all sub-arrays. For $\Psi(\sigma,\tau)$ if %the size of the last sub-array is $2^k$. Then probability of event $A$ is given by $P[A]$ as follows. 
%\[ P[A] < P[\cap^{k-1}_{u=1}\cap^u_{v=1}A^u_v] \]
%Existence of an unmatched index is complement of event $A$. Hence $1-P[A]$. Thus on simplification we get
%\[ P[A] < \frac{1}{2^{\sum_{u=0}^{k-1} u^2}} \]
%\[ P\left[\text{RICE$_i$, RICE$_j$ $(i > j)$ have unmatched index}\right]
%\geq 1-\frac{1}{2^\lambda}  \] 
%\[ \text{where } \lambda = \sum_{u=0}^{k - 1}u^2 = \frac{(k-1)k(2k - 1)}{6} \] 
%\end{proof}

\begin{definition} {\bf Strong Unmatched Index.} An index in RICE$_i$ is called {Strong Unmatched Index} if it is an unmatched index $\forall$ RICE$_j$ where $j<i$.
\end{definition}

We now evaluate the distribution of number of strong unmatched index
in RICE$_i$. The presence of even a single strong unmatched index
implies that even if an adversary assists a quasi-honest node in a
free-loading attack, by revealing root values corresponding to indices in
earlier rounds, these prove insufficient to compute the digest of
RICE$_i$.
% {\bf We need an additional statement or two here. We need to say Prob( X$\ge$ const) is negligible in some func(T) for fixed $i$. Very easy for const=1. The point is that there will be const unmatched indices w.h.p} 
\begin{theorem}
Let $X_k$ denote the number of strong unmatched indices in RICE$_i$ where $\Psi'$ terminates in an segment of size $2^k$. Then $X_k$ is strictly smaller than a {\it Poisson Binomial Distribution} \cite{wang1993number} with mean $\mu = \frac{(k-1)k}{2}-2(i-1)$ and variance $\nu^2 \approx \sum_{n=1}^{k-1}n(1-\frac{i-1}{2^n})\frac{i-1}{2^n}$.
\label{thm:poisson binomial}
\end{theorem}
\begin{proof}
 % Let $X_m$
  %where $m=1,2,2,3,3,3,4\ldots,k$
The occurrence of  a strong unmatched index in a segment of size $2^m$ in RICE$_i$ is
a Bernoulli random variable with mean lower bounded by
$1-\frac{(i-1)}{2^m}$.
This is a tight bound and the event corresponding to the lower bound occurs when all previous rounds have strong unmatched indices in this segment. There are $m$ such segments of length $2^m$.
%Let $X = m \Sigma_m X_m$ be the random variable denoting the number of strong unmatched indices in RICE$_i$.
%
In the case where {\em all} RICE$_j$ have strong unmatched indices in {\em all} segments,
$X$ is random variable with {\it Poisson binomial distribution} \cite{wang1993number} with mean and variance given in the statement of the theorem.
%of $X$ i.e $\mu_X = \sum_m \mu_{X_m}$.
%Since $\Psi(\sigma,\tau)$ terminates in a subarray of size $2^k$, 
%   \[ \mu_X = \sum_{n=1}^{k-1} n\left(1-\frac{i-1}{2^n}\right) \]
%    \[ \mu_X > \frac{(k-1)k}{2}-2(i-1) \]
%Similarly, for the variance of $X$ i.e $\nu_X^2 = \sum_m \nu^2_{X_m}$ and the result follows.
\end{proof}

Given the above we proceed to find a lower bound on the probability
that the number of strong unmatched index $X_k$ is greater than some
$x_k$. We achieve this by finding a lower bound on $P(X \ge
x)$ in the $i^{\rm th}$ round with the tail of a Binomial Random
variable. As CIC size $T\rightarrow \infty$ we prove that for a series
$x_k\rightarrow \infty$, the tail of the binomial and hence $P(X_k \ge
x_k)$ goes to 1. This means that the number of strong indices
increases without bound w.h.p. as $T$ increases. This in turn
reduces the chances of success of a free-loading attack as it becomes
virtually impossible for a node to guess the $root$ values at an increasingly
large set of strongly unmatched indices.

\begin{theorem}
Let $\Psi(\sigma, \tau)$ in the $i^{th}$ round of \miracle\ ends in a segment of size $2^k$. Given $X_k$ defined as in Theorem \ref{thm:poisson binomial}, we can lower bound the tail probability of $X_k$ i.e $P(X_k\ge x_k)$ for any $x_k\le k(k-1)/2$ with the tail probability of $\mathcal{B}\left(\frac{(k-1)k}{2}-\frac{(b_2-1)b_2}{2}, x_k, 1-\frac{i-1}{2^{b_2}}\right)$ for any $b_2 \le k$ where $\mathcal{B}(n,.,p)$ is a binomial distribution with $n$ trials with $p$ as success probability of each trial. 
\end{theorem}
\begin{proofsketch}
Call the occurrence of a strong unmatched index in a segment of size
$2^m$ in RICE$_i$ as a trial in that segment.  The trial is a
Bernoulli random variable with mean lower bounded by
$1-\frac{(i-1)}{2^m}$.  If $m>b_2$ then the mean has lower bound
$1-(i-1)/2^{b_2}$ and if $m\le b_2$ the mean is lower bounded by 0.
Hence $X_k$ which is the sum of all trials has tail distribution
strictly higher than the tail of the sum of $(b_2-1)b_2/2$ i.i.d. Bernoulli random
variables with mean $1-(i-1)/2^{b_2}$. There are
$\frac{(k-1)k}{2}-\frac{(b_2-1)b_2}{2}$ number of segments with
$m>b_2$.
The result follows.
\end{proofsketch}
 \begin{lemma}   
As $k\rightarrow \infty$, $P(X>\sqrt{k})\rightarrow 1$.
 \end{lemma}
 \begin{proofsketch}
  Choose $x_k=\sqrt{k}$. The result follows from the above Theorem and the use of the well-known bound on the tail distribution of $\mathcal{B}(n,l,p)$ given by
   \[ P\left ( \mathcal{B}(n,.,p)>l \right )\ge 1 - e^{-2\frac{(np-l)^2}{n}} \]
   \end{proofsketch}

\noindent{\bf Remark 1:}
Recall that \miracle\  allows the
system designer to choose an appropriate $|ES|$ size to achieve an
expected number of rounds. In this way, the number of rounds can be
limited to less than a constant $i$
$w.h.p$.

 {\bf Remark 2:}
 Since $\sqrt{k}$ grows unboundedly with $k$, it follows that for large sized ITs, and some finite round $i$, the number of strong indices grows unboundedly w.h.p.
 Since the $root$ values at these strong indices are not known w.h.p. (except for trivial CICs where storage does not change over indices) the final $seed$ also cannot
 be known w.h.p.

 The next result shows that the probability of occurrence of any particular seed value is vanishingly small assuming the roots at different strong indices
 are mutually independent.
\begin{theorem}
  Let the probability mass distribution of the $root$ at all strongly unmatched indices in round $i$ be upper bounded by $1- \lambda$ for some $\lambda>0$. Let the last segment
  of the CIC be of size $2^k$. Then as $k\rightarrow \infty$ the probability mass function of the $seed$ at the end of RICE$_i$ is negligible assuming an ideal hash function, and that the root at different unmatched indices are mutually independent.
\end{theorem}
\begin{proofsketch}
  Let $j\in [1:X]$ denote the $X$ strongly unmatched indices, and $r_j$ and $s_j$ the corresponding root and seed. Since the hash function is ideal, it maps unique inputs to unique outputs.

  Thus $P({\sf hash}(s_j||r_j))=P(s_j,r_j)=P(s_j)P(r_j)$. The last equality is due to the independence assumption.  We assume that $s_1$, the seed at the first strong unmatched index, is known to the node and hence $P(s_1)=1$.
  Denoting $seed_{X+1}$ as the final seed, we have $P(seed_{X+1})=\prod_j P(r_j)\le \prod_j \max P(root_j)\le (1-\lambda)^X$. Since $X$ is larger than $\sqrt{k}$ w.h.p. as $k\rightarrow \infty$ we have $P(s_{X+1})\rightarrow 0$.  
\end{proofsketch}
\\

\noindent{\bf Remark:} The roots of indices ``far apart'' being independent is not unrealistic, except for trivial CICs. Strong unmatched indices are in different segments and hence except for neighboring indices, they are separated by whole segments, and hence  we conjecture that the independence assumption is a good approximation in practice. We also conjecture that the same result holds for weaker assumptions than stated in the theorem and leave the proof for future work.

\subsection{Free-loading attack}
\label{sub:freeloading}
We now analyze a free-loading attack where a quasi-honest node skips 
computation of the CIC by using information available on the 
blockchain and/or state information of $\Psi(\sigma, \tau)$ from 
previous rounds received from an adversary~\textsection\ref{sec:rice}.
We consider the best case scenario for the free-loading node where it 
knows the correct $root$ of the $digest$ w.h.p. but has to guess the $
seed$. We analyze the case where Byzantine nodes have maximum fraction 
$f_{max}$ in SP and all submit the same incorrect $root$ with the same 
$seed$ in order to maximize the probability of \miracle\ selecting 
their solution, and where quasi-honest nodes do not collude.

Denote the profile where all quasi-honest nodes execute the CIC as 
$\overrightarrow{a}$ and the profile where only a single quasi-honest 
node $n_i$ free-loads as $\overrightarrow{a}_{-i}$. With 
$\overrightarrow{a}$ the analysis of \miracle\ with honest and 
Byzantine nodes holds. 
Hence quasi-honest nodes win reward $\mathcal{R}$ with 
probability $1-\beta$, and lose their deposits $\mathcal{D}$ with 
probability $\beta$. The cost of computing the CIC is $c_1$.
Hence the utility for $n_i$ with this profile is
\begin{equation}
	\mathcal{U}_i(\overrightarrow{a}) = (1-\beta)\mathcal{R} - \beta \mathcal{D} - c_1
    \label{eq:utility a1}
\end{equation}

Let $\gamma$ be the probability of $n_i$ guessing the correct seed
while free-loading. If it guesses the correct seed then its
probabilities of winning a reward and losing its deposit are
$(1-\beta)$ and $\beta$ as above. If it guesses the wrong seed then it
loses its deposit. We denote by $c_2$ the cost of bandwidth consumed 
for downloading intermediate $storage$ of previous rounds from an
adversary and analyzing them to predict the $seed$. Then the utility
\begin{equation}
	\mathcal{U}_i(\overrightarrow{a}_{-i}) = \gamma((1-\beta)\mathcal{R} - \beta \mathcal{D})-(1-\gamma)\mathcal{D} - c_2
    \label{eq:utility a2}
\end{equation}
From
(\ref{eq:utility a1}) and (\ref{eq:utility a2}) we obtain $ \mathcal{U}_i(\overrightarrow{a})-\mathcal{U}_i(\overrightarrow{a}_{-i})>0$ iff
\begin{equation}
  \mathcal{R} + \mathcal{D}> \frac{c_1-c_2}{(1-\beta)(1-\gamma)} \approx c_1 - c_2
    \label{eq:utility diff}
\end{equation}
where the last approximation is due to the fact that $\gamma$ is vanishingly small in practice  and $\beta$ is a design parameter chosen to be small. Since $\mathcal{R}>c_1$, that is the reward must be more than the cost of computation, we see that (\ref{eq:utility diff}) is true. Hence profile $\overrightarrow{a}$ is a Nash equilibrium \cite{nash1950equilibrium}.

\subsection{Collusion Attack}
\label{sub:collusion}
We now  consider the case where a group $\mathcal{C}$ of ES nodes 
collude to submit a common seed. We assume they know the correct root $
w.h.p$, that Byzantine nodes all submit the same incorrect $root$ with 
the same $seed$ in order to maximize the probability of \miracle\ 
selecting their solution, and that all other quasi-honest nodes 
execute the CIC correctly.
Since $|ES|$ is random, nodes in $\mathcal{C}$ cannot be entirely sure 
if $|\mathcal{C}|$ is larger than $th_1 |ES|$ or less than $th_2 |ES|$.
Suppose $|\mathcal{C}|>th_1|ES|$ with probability $\gamma_1$ and 
$|\mathcal{C}|<th_2 |ES|$ with probability $\gamma_2$.
The computation cost of colluding requires solution of ZK-proofs since 
nodes need to prove they belong to ES without revealing their 
Denote the associated costs by $c_3$ and this profile by 
$\overrightarrow{a}_{-\mathcal{C}}$.

In case the Byzantine nodes win \miracle, $\mathcal{C}$ lose their deposits. In case the correct root is selected by \miracle, $\mathcal{C}$ win a reward with probability $\gamma_1$, and lose their deposits with probability $\gamma_2$. Hence utility for node $n_i\in \mathcal{C}$
\begin{equation}
	\mathcal{U}_i(\overrightarrow{a}_{-\mathcal{C}}) = \gamma_1((1-\beta)\mathcal{R} - \beta \mathcal{D})-\gamma_2\mathcal{D} - c_3
    \label{eq:utility col}
\end{equation}
    In case $\gamma_1=1$, $U_{i}(\overrightarrow{a})$ is a $\epsilon$-Nash equilibrium \cite{radner1982collusive} with $\epsilon = c_3-c_1$. In this special case, if the $c_3$ is larger than the CIC computation cost itself, the nodes are better off being honest.
Note that  higher $|\mathcal{C}|$ increases $\gamma_1$, but also increases $c_3$ because of more ZK proofs, and related communication costs.

\subsection{Meeting the Requirements.}
\label{sub:meeting requirements}
%{\bf SOURAV: We can remove the following subsection by mentioning these points prior to RICE and \miracle\ analysis
%
%About drawing link between the analysis given we may need to write a paragraph}
We here summarize how various mechanisms in YODA meet the
goals listed in~\textsection \ref{sec:introduction}.  Most
requirements are met due to \miracle. For all $f_{max}<0.50$, \miracle
terminates and thus off-chain execution of a CIC also {\em
  Terminates}. \miracle\ allows a system parameter $\beta$ which is the
probability of accepting a incorrect solution, thus achieving {\it
  Validity} with tunable high probability. {\it Agreement} on
off-chain CIC execution trivially follows from the {\it safety}
guarantees of the blockchain.  Recall in YODA, to initiate the reward
mechanism process ES nodes need to submit the correct storage root in
the blockchain and the miners verify it before its inclusion ensuring
{\em Availability} of the post-execution state $\sigma^*$.

Since YODA never requires  the blockchain to verify the CIC execution on-chain, YODA is {\it Oblivious}. 
YODA requires ES to be much smaller than naive approaches discussed in \textsection \ref{sec:miracle} thus making YODA {\em Efficient}. 
\miracle\ is {\it Adaptive} to the fraction of Byzantine nodes in SP, and terminates earlier the smaller this is  \textsection \ref{sec:miracle}. For appropriate choice of $th_1, th_2$ in \textsection \ref{sub:reward} YODA ensures {\it weak-Fairness}.
% {\bf NOT SURE WHAT YOU MEAN BY MAXIMUM} In situation where $\mathcal{U}(a_1)$ is maximum YODA always ensures {\it weak-Fairness}.

\subsection{Stronger Adversary Scenarios with Rational Nodes}
\label{sub:rational nodes}
% {\bf COLLUSION??}
% {\bf Pasting from WhatsApp}
% [7:00 PM, 8/1/2018] Sourav Das Yoda: We discussed that due to commit and reveal mechanism the adversary won't be able to prove anything strong about the intermediate states in the current round without revealing his commitment i.e sortition result.
% {\bf  Sourav Das Yoda: With the mechanism that we discard the commitments whose sortition results are already known, they won't contribute in the computation of $n^2$ for that particular nonce. The for a node to follow the adversary he will face uncertainties about like correctness of nonce.
%  Sourav Das Yoda: Another approach we can do to disincentivize nodes from following adversarial release is the idea that we punish nodes which submit random nonce (as discussed earlier if only a few of them submits that particular nonce). Thus by discarding the adversarial commitments, we reduce the count on the nonce committed by the adversary.
%  Sourav Das Yoda: These incentives schemes might work but we don't have any strong arguments about their correctness.
%  Sourav Das Yoda: One way of handling it would be "We prove correctness with making stronger assumptions like no information from current rounds". We give some incentive schemes to realize our assumptions but since we don't prove them we don't claim about their equivalence with system with fewer assumptions}
We now discuss how YODA may perform if we replace quasi-honest nodes by {\em rational} nodes and also allow stronger adversaries. Unlike quasi-honest nodes, rational nodes are not conservative 
(ref.~\textsection\ref{sub:threat assumption}) towards Byzantine nodes. 
The stronger adversary is allowed to try to reveal information about 
the current round to a rational node, and not just information about 
past rounds.

%Let's consider how things changes when we weaken our assumption and consider the system with {\it rational} nodes (that tries to maximize their {\it Utility}) along with a stronger adversary who can also communicate intermediate states from the ongoing round. Let $\mathcal{N}$ be th rational node. Let us also weaken our assumption that $\mathcal{N}$ is conservative towards Byzantine nodes.  
%% Additionally to make the system more challenging, lets allow the rational nodes to prove about their selection to an Adversary by using {\it Interactive Zero-knowledge proofs}. We rule out the usage of Non-Interactive Zero-knowledge proofs because a mechanism can be built where an Adversary can use this proof to convince the blockchain about the nodes selection in ES and thus leading to forfeiting of deposits. Let us also consider the usage of Smart-contracts to put down rules for agreement among nodes for criminal activities as shown in \cite{juels2016ring}.
%Let's elaborate on methods how $\mathcal{N}$ use these techniques to cheat/attack in YODA be it either predicting the $digest$ with large probability or making \miracle\ accept an incorrect solution. Lets consider them one-by-one starting with the case of predicting the correct $digest$.

Consider the case where an adversary $\mathcal{A}$ is actively
providing information about the $root$ at intermediate RICE indices
for the current round. Recall from RICE \textsection
\ref{sub:rice analysis} that for each round we will have a large number
of strong unmatched indices with large probability, and thus for all
such unmatched indices, a rational node $\mathcal{N}$, has to obtain
the $root$ from $\mathcal{A}$.
%For round $i$ of a CIC execution let $U$ be the set of unmatched indices i.e $U= \{u_1,u_2,\ldots, u_{|U|}\}$. For each $u_i \in U$ if the storage of the CIC changes in $u_{i-1}$ to $u_i$ the root will be indistinguishable from a random string without its pre-Image. Hence $\mathcal{N}$  will have to collect information about corresponding pre-Images. 
How will $\mathcal{A}$ convince $\mathcal{N}$ about the correctness of
the states at these intermediate points? One mechanism is that 
$\mathcal{A}$ commits its digest in the current round and gives a
zk-Proof to $\mathcal{N}$ about it. A challenge for $\mathcal{N}$ is 
that it does not know whether $\mathcal{A}$ is Byzantine or rational.  
Even if we consider that rational nodes commit correct digest, when 
$\mathcal{A}$ is Byzantine then its commitment could be false.
% {\bf removed discussion about chain forking attack by rational node which was too hard to understand.}
%Even a rational $\mathcal{A}$ can
%commit to a incorrect $seed$ to convince $\mathcal{N}$ to release
%an identical commitment and then possibly creating a alternate chain
%where it switch its commitment to the correct digest whereas including
%the incorrect digest of the $\mathcal{N}$.  $\mathcal{N}$ can use
%smart contracts as shown in \cite{juels2016ring,velner2017smart} just
%to ensure that if $\mathcal{A}$s' claims are incorrect $\mathcal{A}$
%pays a penalty, but no smart contract solution works when
%$\mathcal{A}$ is Byzantine.

Consider the case where two or more rational nodes in an ES want to
collaborate to generate the correct $digest$. Consider one rational 
node $\mathcal{N}_R$ interacting with another node $\mathcal{N}_U$ 
which might potentially be Byzantine. First, both nodes must prove 
that they belong to ES. Ruling out a non-Interactive zk-Proof, 
$\mathcal{N}_R$ and $\mathcal{N}_U$ have to produce an interactive 
zk-Proof, and this may be vulnerable to DoS attacks.
Once the proofs are established, one strategy could be to split the 
execution among them. This situation again boils down to trusting the 
result claimed by $\mathcal{N}_U$ about its execution.

Finally we  discuss pragmatic concerns related to a {\it collusion attack}. In our utility analysis we considered the case
where all nodes in the colluding set $\mathcal{C}$ are all
quasi-honest but in practice it will possibly contain Byzantine
nodes as well. This will possibly lower the probability $\gamma_1$ of
successful collusion. Since the success of the collusion attack depends
on $|ES|$, imperfect knowledge of $|ES|$ during the commitment phase 
of RICE lowers $\gamma_1$ even more. 
Interestingly even with knowledge of $|ES|$ prior to collusion and 
$|\mathcal{C}|\ge th_1|ES|$, the success truly depends on the behavior 
of Byzantine nodes in $\mathcal{C}$.

\section{Implementation and Evaluation}
\label{sec:evaluation}
\begin{figure*}[!htb]
\minipage{0.33\textwidth}
  \pgfplotsset{footnotesize,height=6cm, width=\linewidth}
  \begin{tikzpicture}
    \begin{semilogxaxis}[
      ylabel=Time (in seconds),
      xlabel=Gas Usage (multiples of {$5.3$}){,} {$|ES|=40$},
      bar width=12pt,
      legend pos=north west,
    ]
    \addplot [
      ybar stacked,
      black,
      fill=clr1,
    ]
    coordinates {(10^6,60) (10^7,60) (10^8,60) (10^9,60) (10^10, 160)}; 

    \addplot [
      ybar stacked,
      black,
      fill=clr2,
    ]
    coordinates {(10^6,20) (10^7,20) (10^8,20) (10^9,20) (10^10,20)}; 
    \addplot [
      ybar stacked,
      black,
      fill=clr3,
    ]
    coordinates {(10^6,40) (10^7,40) (10^8,40) (10^9,40) (10^10,40)}; 

    \addplot [line width=0.30mm, red, mark=square] 
    coordinates {(10^6,20) (10^7,20) (10^8,21) (10^9,32) (10^10, 82)};
    \addlegendentry{$w_{src}$}
    \addlegendentry{$w_{buf}$}
    \addlegendentry{$w_{sr}$}
    % \addlegendentry{Total time}
    \end{semilogxaxis}
  \end{tikzpicture}
  \caption{Measured CIC execution time with varying gas usage.}
  \label{fig:gasincrease}
\endminipage\hfill
\minipage{0.33\textwidth}
    \pgfplotsset{footnotesize,height=6cm, width=\linewidth}
    \begin{tikzpicture}
    \begin{semilogxaxis}[
    xtick=\empty,
    xlabel=\# Parallel CICs{,} {$|ES|=40$},
    ylabel=Time (in seconds),
    extra x ticks = {1,2,4,8,16},
    extra x tick labels = {1,2,4,8,16},
    bar width=12pt,
    legend pos=north west
    ]
      \addplot [
        ybar stacked,
        black,
        fill=clr1,
      ]
      coordinates {(1,60) (2,60) (4,60) (8,60) (16, 60)}; 
        \addplot [
        ybar stacked,
        black,
        fill=clr2,
      ]
      coordinates {(1,20) (2,20) (4,20) (8,20) (16, 20)}; 

      \addplot [
        ybar stacked,
        black,
        fill=clr3,
      ]
      coordinates {(1,40) (2,40) (4,40) (8,40) (16, 40)}; 
      \addplot [line width=0.30mm, red, mark=square, error bars/.cd, y dir=both,y explicit] coordinates {(1,20) (2,21) (4,26) +- (4.5,4.5) (8,31.34) +-(11.21,11.87) (16, 45.06) +- (26.13,22.12)};
    \addlegendentry{$w_{src}$}
    \addlegendentry{$w_{buf}$}
    \addlegendentry{$w_{sr}$}
    \end{semilogxaxis}
    \end{tikzpicture}
    \caption{Average digest commit time with increasing number of parallel ITs.}
    \label{fig:multiple}
\endminipage\hfill
\minipage{0.33\textwidth}
\pgfplotsset{footnotesize,height=6cm, width=\linewidth}
  \begin{tikzpicture}
  \begin{axis}[
    xmin=0.20,
    xmax=0.45, 
    ymin=0.0,
    ymax=8,
    domain=0.20:0.45,
    samples=25,
    xlabel={$f_{max}$},
    ylabel={Expected \# rounds},
    xtick=\empty,
    legend pos=north west,
    extra x ticks={0.2,0.25,0.30,0.35,0.40,0.45},
    extra x tick labels={0.2,0.25,0.30,0.35,0.40,0.45}
  ]
    \addplot[blue, dashed, mark options=solid, mark=triangle] coordinates {(0.20,1.00) (0.25,1.00) (0.30,1.12) (0.35,1.25) (0.40,2.13) (0.45,7.61)};
    \addplot[blue] (x, {((1-10^(-10))*ln((1-10^(-10))/10^(-10))+(10^(-10))*ln((10^(-10))/(1-10^(-10))))/((200*(1-2*x)*(1-2*x) + (1-2*x)*0.875)/(1.75*x) + 0.5*ln((x/(1-x))))});

    \addplot[red, dashed, mark options=solid, mark=square] coordinates {(0.20,1.00) (0.25,1.00) (0.30,1.06) (0.35,1.21) (0.40,1.47) (0.45,4.52)};
    \addplot[red] (x, {((1-10^(-6))*ln((1-10^(-6))/10^(-6))+(10^(-6))*ln((10^(-6))/(1-10^(-6))))/((200*(1-2*x)*(1-2*x) + (1-2*x)*0.875)/(1.75*x) + 0.5*ln((x/(1-x))))});

    % \addplot[black, dashed, mark options=solid, mark=o] coordinates {(0.20,1.00) (0.25,1.00) (0.30,1.00) (0.35,1.00) (0.40,1.01) (0.45,3.01)};
    % \addplot[black] (x, {((1-10^(-3))*ln((1-10^(-3))/10^(-3))+(10^(-3))*ln((10^(-3))/(1-10^(-3))))/((200*(1-2*x)*(1-2*x) + (1-2*x)*0.875)/(1.75*x) + 0.5*ln((x/(1-x))))});

    % \addlegendentry{$\beta=10^{-3}$ (Theo.)}
    % \addlegendentry{$\beta=10^{-3}$ (Expt.)}
    \addlegendentry{$\beta=10^{-10}$ (Theo.)}
    \addlegendentry{$\beta=10^{-10}$ (Expt.)}
    \addlegendentry{$\beta=10^{-6}$ (Theo.)}
    \addlegendentry{$\beta=10^{-6}$ (Expt.)}
  \end{axis}
  \end{tikzpicture}
  \caption{Change in Expected number of rounds vs. $f_{max}$ for $|SP|=1600, q=0.125$.}
  \label{fig:betavary}
\endminipage\hfill
\end{figure*}
To experimentally evaluate the scalability of YODA, we have 
implemented a prototype which includes all parts of YODA except SP 
selection, on top of the popular Ethereum geth client version 
{\tt 1.8.7} and evaluate them in a private network. The SP selection procedure is independent of CIC 
transaction gas limits and hence does not impact scalability. Our 
implementation consists of $\sim$500 lines of code (LOC) in Solidity 
(for the Master Contract and sample CICs) and $\sim$2000 LOC in python 
on top of the interface for the modified client whose task is detailed 
below. 

\noindent{\bf Experimental Environment.}
We use 8 physical machines each with a $2.80\times8$ GHz  intel 
Core-i7 processor with 8GB RAM and running Ubuntu 17.10 to simulate 16 
off-chain clients. Each client emulates 100 YODA nodes thus making 
$M=|SP| = 1600$. Since off-chain CIC execution requires considerable 
computations resources, we restrict the number of clients per machine 
to two. All these 16 clients are connected to an Ethereum network
created using geth which we consider to be  the blockchain network. 
The blockchain runs on Proof of Authority in geth which is a developer 
mode option in the Ethereum private network to keep the block mining time 
fixed. A separate physical machine with similar specifications  
handles transactions and mines blocks. To achieve a required $|ES|$ we
allow a single client to execute an IT only once and submit commitments 
for multiple SP nodes based on how many of its nodes are selected by 
Sortition. For each execution we regulate $q$ to achieve an 
appropriately sized ES.

\noindent{\bf Executing CICs in the EVM.}
To execute ITs off-chain, each client in our system runs a modified
EVM supported with an implementation of RICE. This EVM provides the
interface for off-chain execution of a CIC. The clients deploy the
contract in their local EVM whenever a CIC is created on the blockchain.  When a IT is broadcast on-chain, clients download the transaction and execute them locally in the EVM based on the sortition result.  After execution each client broadcasts a transaction containing the $digest$ for each of its ES nodes.  An identical mechanism is used to reveal the
commitment.

\subsection{Scalability of CICs.}
\label{sub:cic scalability}
We start with a CIC containing a parameterizable function called
\texttt{Compute()} with parameter $\eta$ which sets the amount of gas
to be used. Internally \texttt{Compute()} runs iterations modifying a
CIC state variable after performing arithmetic operations.  We start
\texttt{Compute()} with gas usage $5.3\times10^6$ and increase this by 
a factor of 10 rising up to $5.3\times10^{10}$.  Note that each
Ethereum block can only accommodate up to 8M gas in 15 seconds
(avg. block generation time) at the time of writing and hence we
conclude that its maximum gas usage per second is ${5.3\times10^5}$.

\noindent{\bf Varying Gas Usage.} 
In Figure~\ref{fig:gasincrease} we plot the average of measured time
from the start of the round up to the commitment (S1-S2) as a red line
and the range of values as error bars. Observe that the average time 
remains constant at about 20 secs till gas usage of $5.3\times10^8$ 
after which it increases to about 30 secs for $5.3\times10^9$ and 
to about 100 secs after that. The total submission time includes
off-chain computation time and in addition Ethereum on-chain
delays, such as the time for a transaction to be included in a block
on-chain. Clearly for gas $5.3\times10^9$ or less the on-chain delays
dominate after which off-chain delays dominate. 

We also plot a bar graph depicting the windows for committing the 
storage root $w_{src}$, the buffer period $ w_{buf}$, and the window 
to reveal the storage root value $w_{sr}$. Among the three only 
$w_{src}$ depends on $gas$ usage since the computation of the IT 
happens during this time. The total time for one round is 
$w_{src} + w_{buf} + w_{sr}$.

%and also the average total round time (S1-S6) for CIF execution with increasing gas usage.  
%In our experiment we kept $w_{buf}=2$ blocks which is equivalent to $30$ seconds and $w_{sr} = 6$ blocks or $90$ seconds.
%In Figure \ref{fig:gasincrease} the red line below reflects the mean time at which nodes commits their storage root in the blockchain.
%Total time required for off-chain computation is duration $w_{src} + w_{buf} + w_{sr}$ among which only $w_{src}$ depends on $gas$ usage thus result in increase in time.
From the experiment with CIC gas usage equal to $5.3\times10^{10}$, 
we see that YODA consumes $240$ Million gas per second. This amount 
is $450\times$ more than the existing amount of gas Ethereum can use 
per second. Note that this speedup is when only a single IT is running. 
With parallel execution of ITs this scales up further as we 
demonstrate in our next experiment.

\noindent{\bf Parallel CICs.} We further test YODA by running up to 
16 parallel ITs. Figure~\ref{fig:multiple} shows the time taken for 
executing different number of concurrent ITs. All ITs are invoked at 
once in a single block on-chain and the $gas$ usage of all are kept 
identical.
The red line in Figure~\ref{fig:multiple} records the average of the 
storage root commitment times and error bars are used to indicate 
the range of these. Observe that  the minimum commit time remains 
almost constant, indicating that the time for off-chain execution is 
the same. However the maximum value increases. This is because more 
blocks are needed to include the increased number of commitment 
transactions. As a result the average commit time increases gradually. 
As future work we will devise mechanisms to automatically provision 
$w_{src}$ taking this phenomenon into account.

% \begin{figure}[h!]
% \centering
% \includegraphics[width=\linewidth]{t_parallel_cif.png}
%   \caption{Time with running multiple simultaneous CIF}
%   \label{fig:multiple}
% \end{figure}

%\noindent{\em Comparison with Other Systems} Contemporary distributed ledgers which support Turing-complete smart contracts are Hyperledger, Ripple, Ethereum, and TrueBit \cite{cachin2016architecture, schwartz2014ripple, buterin2014next, teutsch2017scalable}.
%Hyperledger and Ripple are permissioned distributed system hence making comparison with YODA which is permissionless inappropriate.
%Since TrueBit does not has any implementation we only compare YODA qualitatively with Ethereum.

%[DO WE NEED THIS FIGURE? WE HAVE ALREADY SAID THAT YODA IS 300 TIMES FASTER THAN ETHEREUM. ALSO HOW IS ETHEREUM ABLE TO USE 100x THE  MAXIMUM AMOUNT OF GAS IT CAN USE?]
%We compare YODA and Ethereum based on {\em gas usage per second}.
%Figure \ref{fig:ethcompare} depicts the time taken by both system with increasing amount of gas. Observe that, time required for Ethereum increases linearly whereas in YODA it remains almost constant.
%Additionally, the figure also depicts that for Lighter functions are better off executed on-chain and will less require time to execute.

% \begin{figure}[h!]
% \centering
%   \includegraphics[width=\linewidth]{eth_vs_yoda.png}
%   \caption{Time requirement comparison with Ethereum}
%   \label{fig:ethcompare}
% \end{figure}

\noindent{\bf Evaluation of \miracle.}
We next evaluate the performance of \miracle\ in the presence 
of a Byzantine adversary. In our experiments the adversary uses 
the best strategy, that is it submits a single incorrect solution 
for all nodes it controls.

For system design the expected number of rounds is a crucial
parameter. We determined this quantity experimentally and compared it
to its theoretical approximation.  In Figure \ref{fig:betavary} we
plot $\mathbb{E}[\text{\# Rounds}]$ versus the fraction of
Byzantine nodes $f$ for different values of the probability of
accepting an incorrect storage root $\beta$ in the range $10^{-3}$ to
$10^{-10}$. For this experiment we fixed parameter values $q=0.125,
M=1600$ giving us an ES of expected size 200.  Notice that
$\mathbb{E}[\text{\# Rounds}]$ increases with an increase in $f$ and
largely agrees with the theoretical approximation. The theoretical
approximation has an artifact in that it can be less than 1 which is
impossible because the number of rounds is at least 1 always.

\noindent{\bf Evaluation of RICE.}
% \vinay{Not clear what we are doing here. Are we putting RICE on plain Ethereum? Or is this YODA without RICE vs YODA with RICE?}
% \sourav{We evaluate RICE only using EVM (without YODA)}
We next evaluate the overheads associated with RICE when implemented 
on the EVM geth client. For this experiment, we measure CIC execution 
time on the unmodified EVM and then perform the same experiment 
in a EVM modified with RICE implementation. For each gas usage, we aggregate
the results over 200 repetitions. Figure \ref{fig:rice absolute} shows the 
time difference of CIC execution with RICE and without RICE. As expected, 
the absolute difference increases as gas increases due to the presence 
of more update indices. More interesting to observe is Figure~\ref{fig:rice relative} 
where we plot gas usage vs. relative execution time i.e ratio of absolute 
time difference and CIC execution time without RICE.
% $\frac{\text{Time(with RICE) -Time(without RICE)}}{\text{Time(without RICE)}}$.
First observe that the relative overhead due to RICE is extremely low. 
As gas increases, the relative time decreases because the RICE indices 
become sparse in later segments and hence add less overhead. 
During the early part of the graph wee see a small aberration. 
This is because time with and without RICE are small and hence minor 
variations in the absolute time difference get magnified 
relative to time without RICE.
%But as we increase gas usage, the graph stabilizes and it is apparent that the relative overhead decreases.
\begin{figure}[h!]
    % \minipage{\linewidth}
    % \centering
    % \begin{subfigure}{\linewidth}
    \centering
    \pgfplotsset{small,height=4cm, width=\linewidth}
    \begin{tikzpicture}
    \begin{semilogxaxis}[
    legend pos=north west,
    % xmin=-1,
    % ymin=-1,
    % ymax=10,
    % xtick=\empty,
    % ytick=\empty,
    % extra x ticks={0,1,...,9},
    % extra x tick labels = {0,1,...,9},
    % extra y ticks={0.0,2,...,10.0},
    % extra y tick labels={0,2,...,10},
    xlabel=Gas usage(in Millions),
    ylabel=Time(in seconds),
    grid=major
    ]    

      \addplot [line width=0.25mm, blue, mark=square] coordinates {( 1.086869 , 0.001792359109999996 ) ( 2.146869 , 0.002796019980000011 ) ( 4.266869 , 0.0033971302099999575 ) ( 8.506933 , 0.004937859729999957 ) ( 16.986869 , 0.006934588240000096 ) ( 33.946869 , 0.014554146069999774 ) ( 67.866869 , 0.03029743274999987 ) ( 135.706869 , 0.0701978814099989 ) ( 271.386869 , 0.15428714813999977 ) ( 542.746869 , 0.2749215859300002 ) ( 1085.466869 , 0.44562336938000224 ) ( 4341.786869 , 1.4559698687000056 ) ( 8683.546869 , 2.4471964238000217 ) };
    \end{semilogxaxis}
    \end{tikzpicture}
    \caption{Absolute overhead of RICE in CIC execution plotted against increasing gas usage.}
    \label{fig:rice absolute}
\end{figure}
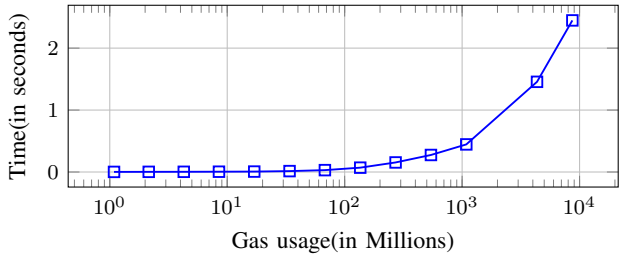
    % \endminipage\hfill
    %
    % \begin{subfigure}{\linewidth}
    % \minipage{\linewidth}
\begin{figure}[h!]
    \centering
    \pgfplotsset{small,height=4cm,width=0.95\linewidth}
    \begin{tikzpicture}
    \begin{semilogxaxis}[
    xlabel=Gas usage(in Millions),
    ylabel={$\frac{\text{Absolute time difference}}{\text{Time without RICE}}$},
    legend pos=north west,
    ymax=0.25,
    ytick=\empty,
    extra y ticks={0,0.05,...,0.25},
    extra y tick labels = {0,0.05,0.10,0.15,0.20,0.25},
    grid=major
    ]
      \addplot [line width=0.25mm, blue, mark=square] coordinates {( 1.086869 , 0.08996924130235726 ) ( 2.146869 , 0.07808681276070428 ) ( 4.266869 , 0.049294443145826136 ) ( 8.506933 , 0.036947388648926154 ) ( 16.986869 , 0.026281226613264427 ) ( 33.946869 , 0.027617804146803683 ) ( 67.866869 , 0.028914140694510093 ) ( 135.706869 , 0.033917838741313745 ) ( 271.386869 , 0.03673828617808041 ) ( 542.746869 , 0.03279082289810635 ) ( 1085.466869 , 0.026494561779651725 ) ( 4341.786869 , 0.021601461778338417 ) ( 8683.546869 , 0.01812163601123794 )};
    \end{semilogxaxis}
    \end{tikzpicture}
    \caption{Relative overhead of RICE in CIC execution plotted against increasing gas usage.}
    \label{fig:rice relative}
    % \end{subfigure}
    % \endminipage\hfill
    % \caption{Evaluation of RICE implemented in EVM}
\end{figure}
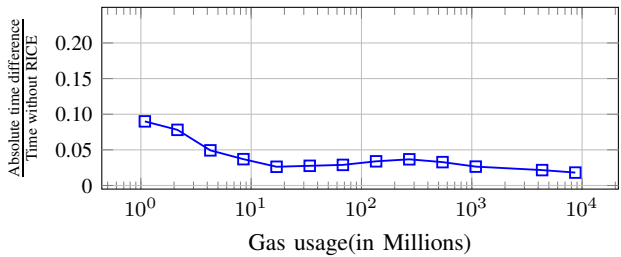

 \section{Related Work}
 \label{sec:related}
% \begin{table*}[h!]
%   \centering
%   \begin{tabular}{|p{1.4cm}|p{3cm}|p{3cm}|p{2.5cm}|p{1.8cm}|p{2.2cm}|}
%   \hline
%   System & CIC Computation & Repetition & Inter-contract & Interactive & Verifier's Reward \\
%   \hline
%   Ethereum & No & Entire network & Yes & No & 0 \\
%   Elastico & No & O(|Committee|) & No & Yes & Uniform \\
%   Chainspace & No & O(|Committee|) & Yes & Yes & Uniform \\
%   Truebit & Yes & Unknown & No & On challenge & Unknown \\
%   YODA & Yes & O(|ES|) & Yes* & No & Uniform \\
%   \hline
%   \end{tabular}
%   \caption{Comparison of YODA with various blockchain systems. *On-chain contracts in YODA can write to each other's storage but not the off-chain contracts.}
%   \label{tab:1}
% \end{table*}

 The threat model of combining Byzantine and selfish nodes in
 distributed systems dates back to Aiyer et al.  \cite{aiyer2005bar}
 Prior to their work threat models in distributed systems either
 considered the presence of perfectly honest nodes and Byzantine
 nodes, or only selfish nodes. The {\it Byzantine Altruistic Rational}
 (BAR) considered a more realistic scenario combining the two models in
 a permissioned cooperative service with PKI
 \cite{aiyer2005bar}. Our threat model is of the BAR variety.
 
% {\it Byzantine Altruistic Rational} (BAR) dates back to \cite{aiyer2005bar} where a try to model a permissioned distributed data service in the presence of PKI. In BAR model of distributed systems non-Byzantine nodes in the system are considered to be purely driven by incentives to increase their Utility. 

Most analysis of blockchain consensus protocols in the permissionless
case limit themselves to a threat model of Byzantine with honest
nodes which are not BAR models. These include works on 
Fruitchains~\cite{pass2017fruitchains}, Algorand~\cite{gilad2017algorand}, 
and the sleepy model of consensus~\cite{pass2017sleepy}. 
In the sleepy model, honest nodes may go offline and not participate 
in the protocol. Ouroboros~\cite{kiayias2017ouroboros} introduces 
$\epsilon$-Nash Equilibrium for a proof-of-stake protocol. 
Selfish mining~\cite{sapirshtein2016optimal} shows in case a non-zero 
fraction miners in PoW blokchains behave selfishly, honest behaviour is 
no longer an equilibrium, because individuals unilaterally benefit by 
joining hands with the attacker. All these works solve a very different 
problem from YODA, namely that of block consensus. Blocks can potentially 
take many values and are easy to verify. In contrast, CIC computation can 
have only one correct value and are computationally intensive to verify.

Truebit is a proposal to enable CICs on permissionless blockchains, in
the presence of selfish nodes \cite{teutsch2017scalable}.  Truebit
requires a single {\em Solver} to execute and upload the results of
the transaction, and any number of volunteer verifiers to verify the
Solver's solution. There is no bound on the number of verifiers unlike
in YODA. Moreover Truebit does not claim to provide guarantees for
probability of correct CIC computation under a threat model. In fact,
recent work shows that Truebit is susceptible to a Particpation
Dilemma, where if all participants are rational, an equilibrium exists
with only a single verifier which can cheat at will~\cite{kalodner2018arbitrum}.  
It also makes payouts to verifiers rare events, unlike YODA which pays 
ES members rewards immediately.  
% To the best of our knowledge, no implementation of Truebit exists of date and
% it has not appeared at a peer-reviewed venue.

Arbitrum is 
a system for scalable off-chain execution 
of private smart contracts developed concurrently with our work~\cite{kalodner2018arbitrum}. In Arbitrum, each 
smart contract can assign a set of managers who execute its transactions 
off-chain. Any one manager can submit a hash of the updated state on-chain. 
In addition, any other manager can submit a challenge if this earlier 
submitted state is incorrect. 
Arbitrum works under a threat model with at least one
honest manager and the rest of the managers being rational. It has not
been proved to work in the presence of Byzantine managers, or with all
managers being rational. In contrast, YODA works in the presence of
both Byzantine and selfish nodes, none of which need to be
honest. YODA is also not restricted to private smart contracts.

Several other papers focus on sharding for improving performance of
permissionless blockchains \cite{luu2016secure,kokoris2018omniledger,al2017chainspace}. None of these, however, focus on the specific problem of
executing CIC transactions efficiently. They instead increase
throughput in terms of number of non-IT transactions executed over time. The
execution (or verification of correct execution) is implicitly assumed to take
little time, and all miners verify all transactions.

\section{Discussion and Conclusion}
\label{sec:discussion}
We have presented YODA which enables permissionless blockchains to
compute CICs efficiently. 
Experimental results show that individual ITs with gas usage 450 times 
the maximum allowed by Ethereum can be executed using the existing EVM. YODA uses 
various incentives and technical mechanisms such as RICE to force rational 
nodes to behave honestly. Our novel \miracle\ algorithm uses multiple 
rounds to determine the correct solution and shows great savings in 
terms of number of rounds when the actual Byzantine fraction of nodes 
is less than the assumed worst case.

One advantage of YODA is its modular design. Several modules can be left 
intact, while replacing the others. Examples of  such modules are RICE, 
\miracle, SP selection, and ES selection, which can in future be replaced 
by alternatives. 

Several issues need to be addressed before YODA can
become a 
full fledged practical system. 
One open problem we have not addressed is the issue of data. Often, large
CICs are likely to have large state and each IT can potentially modify
many state variable. Broadcasts of every update for every IT can be
costly in terms of communication. %Also, sometime these updates can't
%be accommodated in a single on-chain transaction.
 A possible alternative
to this state update could be storage of data in a Distributed File
System like IPFS~\cite{benet2014ipfs} and using Authenticated Data 
structures such as Versum~\cite{van2014versum} to store a succinct 
representations of it in the blockchain. 
%This will be a interesting future
%research to pursue to realize CICs for large scale applications. 

Another concern is regarding the number of additional transactions
needed to achieve consensus on a CIC. For each round, \miracle\
requires each ES node to submit two short transactions. Also the count
of such transactions depends on $f$. Observe from
Figure~\ref{fig:miracle qpickup} that in the best case YODA requires only
$\approx70$ transactions with $f=0$ and $\beta=10^{-20}$. With contemporary
blockchain solutions that claims to scale up to $1000$s transaction per
second~\cite{gilad2017algorand,zamani2018rapidchain} these
transactions consume relatively small bandwidth.

The periods $w_{src},w_{sr}$ chosen for execution of CICs off-chain 
will in practice also depend on the number of simultaneous ITs being 
currently executed by YODA. This is because, as the number of 
simultaneous ITs increase, the average CIC workload on any node
increases as well, since
each node may belong to multiple ES sets simultaneously.
As CICs are computationally expensive, the MC must further keep a limit on 
ITs at any instant of time to reduce the maximum load on an ES node. 
%Further, YODA must ensure sufficient transaction space in blocks 
%to include all transactions from ES nodes.

% conference papers do not normally have an appendix

% use section* for acknowledgement
\section*{Acknowledgments}
The authors would like to thank Aditya Ahuja, Cui Changze, Aashish Kolluri, Dawei Li,
Sasi Kumar Murakonda, Prateek Saxena, Ovia Seshadri, Subodh Sharma, and anonymous
reviewers for their feedback on the early draft of the paper.

% trigger a \newpage just before the given reference
% number - used to balance the columns on the last page
% adjust value as needed - may need to be readjusted if
% the document is modified later
%\IEEEtriggeratref{8}
% The "triggered" command can be changed if desired:
%\IEEEtriggercmd{\enlargethispage{-5in}}

% references section

% can use a bibliography generated by BibTeX as a .bbl file
% BibTeX documentation can be easily obtained at:
% http://www.ctan.org/tex-archive/biblio/bibtex/contrib/doc/
% The IEEEtran BibTeX style support page is at:
% http://www.michaelshell.org/tex/ieeetran/bibtex/
\bibliographystyle{IEEEtranS}
\bibliography{IEEEabrv,references}

% Generated by IEEEtranS.bst, version: 1.12 (2007/01/11)
\begin{thebibliography}{10}
\providecommand{\url}[1]{#1}
\csname url@samestyle\endcsname
\providecommand{\newblock}{\relax}
\providecommand{\bibinfo}[2]{#2}
\providecommand{\BIBentrySTDinterwordspacing}{\spaceskip=0pt\relax}
\providecommand{\BIBentryALTinterwordstretchfactor}{4}
\providecommand{\BIBentryALTinterwordspacing}{\spaceskip=\fontdimen2\font plus
\BIBentryALTinterwordstretchfactor\fontdimen3\font minus
  \fontdimen4\font\relax}
\providecommand{\BIBforeignlanguage}[2]{{%
\expandafter\ifx\csname l@#1\endcsname\relax
\typeout{** WARNING: IEEEtranS.bst: No hyphenation pattern has been}%
\typeout{** loaded for the language `#1'. Using the pattern for}%
\typeout{** the default language instead.}%
\else
\language=\csname l@#1\endcsname
\fi
#2}}
\providecommand{\BIBdecl}{\relax}
\BIBdecl

\bibitem{aiyer2005bar}
A.~S. Aiyer, L.~Alvisi, A.~Clement, M.~Dahlin, J.-P. Martin, and C.~Porth,
  ``Bar fault tolerance for cooperative services,'' in \emph{ACM SIGOPS
  operating systems review}, vol.~39, no.~5.\hskip 1em plus 0.5em minus
  0.4em\relax ACM, 2005, pp. 45--58.

\bibitem{al2017chainspace}
M.~Al-Bassam, A.~Sonnino, S.~Bano, D.~Hrycyszyn, and G.~Danezis, ``Chainspace:
  A sharded smart contracts platform,'' \emph{arXiv preprint arXiv:1708.03778},
  2017.

\bibitem{androulaki2018hyperledger}
E.~Androulaki, A.~Barger, V.~Bortnikov, C.~Cachin, K.~Christidis, A.~De~Caro,
  D.~Enyeart, C.~Ferris, G.~Laventman, Y.~Manevich \emph{et~al.}, ``Hyperledger
  fabric: a distributed operating system for permissioned blockchains,'' in
  \emph{Proceedings of the Thirteenth EuroSys Conference}.\hskip 1em plus 0.5em
  minus 0.4em\relax ACM, 2018, p.~30.

\bibitem{ben2014succinct}
E.~Ben-Sasson, A.~Chiesa, E.~Tromer, and M.~Virza, ``Succinct non-interactive
  zero knowledge for a von neumann architecture.'' in \emph{USENIX Security
  Symposium}, 2014, pp. 781--796.

\bibitem{benet2014ipfs}
J.~Benet, ``Ipfs-content addressed, versioned, p2p file system,'' \emph{arXiv
  preprint arXiv:1407.3561}, 2014.

\bibitem{buterin2014next}
V.~Buterin \emph{et~al.}, ``A next-generation smart contract and decentralized
  application platform,'' \emph{white paper}, 2014.

\bibitem{croman2016scaling}
K.~Croman, C.~Decker, I.~Eyal, A.~E. Gencer, A.~Juels, A.~Kosba, A.~Miller,
  P.~Saxena, E.~Shi, E.~G. Sirer \emph{et~al.}, ``On scaling decentralized
  blockchains,'' in \emph{International Conference on Financial Cryptography
  and Data Security}.\hskip 1em plus 0.5em minus 0.4em\relax Springer, 2016,
  pp. 106--125.

\bibitem{douceur2002sybil}
J.~R. Douceur, ``The sybil attack,'' in \emph{International workshop on
  peer-to-peer systems}.\hskip 1em plus 0.5em minus 0.4em\relax Springer, 2002,
  pp. 251--260.

\bibitem{eberhardtzokrates}
J.~Eberhardt and S.~Tai, ``Zokrates-scalable privacy-preserving off-chain
  computations.''

\bibitem{gilad2017algorand}
Y.~Gilad, R.~Hemo, S.~Micali, G.~Vlachos, and N.~Zeldovich, ``Algorand: Scaling
  byzantine agreements for cryptocurrencies,'' in \emph{Proceedings of the 26th
  Symposium on Operating Systems Principles}.\hskip 1em plus 0.5em minus
  0.4em\relax ACM, 2017, pp. 51--68.

\bibitem{juels2016ring}
A.~Juels, A.~Kosba, and E.~Shi, ``The ring of gyges: Investigating the future
  of criminal smart contracts,'' in \emph{Proceedings of the 2016 ACM SIGSAC
  Conference on Computer and Communications Security}.\hskip 1em plus 0.5em
  minus 0.4em\relax ACM, 2016, pp. 283--295.

\bibitem{kalodner2018arbitrum}
H.~Kalodner, S.~Goldfeder, X.~Chen, S.~M. Weinberg, and E.~W. Felten,
  ``Arbitrum: Scalable, private smart contracts,'' in \emph{Proceedings of the
  27th USENIX Conference on Security Symposium}.\hskip 1em plus 0.5em minus
  0.4em\relax USENIX Association, 2018, pp. 1353--1370.

\bibitem{kiayias2017ouroboros}
A.~Kiayias, A.~Russell, B.~David, and R.~Oliynykov, ``Ouroboros: A provably
  secure proof-of-stake blockchain protocol,'' in \emph{Annual International
  Cryptology Conference}.\hskip 1em plus 0.5em minus 0.4em\relax Springer,
  2017, pp. 357--388.

\bibitem{kokoris2018omniledger}
E.~Kokoris-Kogias, P.~Jovanovic, L.~Gasser, N.~Gailly, E.~Syta, and B.~Ford,
  ``Omniledger: A secure, scale-out, decentralized ledger via sharding,'' in
  \emph{2018 IEEE Symposium on Security and Privacy (SP)}.\hskip 1em plus 0.5em
  minus 0.4em\relax IEEE, 2018, pp. 583--598.

\bibitem{luu2016secure}
L.~Luu, V.~Narayanan, C.~Zheng, K.~Baweja, S.~Gilbert, and P.~Saxena, ``A
  secure sharding protocol for open blockchains,'' in \emph{Proceedings of the
  2016 ACM SIGSAC Conference on Computer and Communications Security}.\hskip
  1em plus 0.5em minus 0.4em\relax ACM, 2016, pp. 17--30.

\bibitem{luu2015demystifying}
L.~Luu, J.~Teutsch, R.~Kulkarni, and P.~Saxena, ``Demystifying incentives in
  the consensus computer,'' in \emph{Proceedings of the 22nd ACM SIGSAC
  Conference on Computer and Communications Security}.\hskip 1em plus 0.5em
  minus 0.4em\relax ACM, 2015, pp. 706--719.

\bibitem{micali1999verifiable}
S.~Micali, M.~Rabin, and S.~Vadhan, ``Verifiable random functions,'' in
  \emph{Foundations of Computer Science, 1999. 40th Annual Symposium on}.\hskip
  1em plus 0.5em minus 0.4em\relax IEEE, 1999, pp. 120--130.

\bibitem{nakamoto2008bitcoin}
S.~Nakamoto, ``Bitcoin: A peer-to-peer electronic cash system,'' 2008.

\bibitem{nash1950equilibrium}
J.~F. Nash \emph{et~al.}, ``Equilibrium points in n-person games,'' 1950.

\bibitem{pass2017fruitchains}
R.~Pass and E.~Shi, ``Fruitchains: A fair blockchain,'' in \emph{Proceedings of
  the ACM Symposium on Principles of Distributed Computing}.\hskip 1em plus
  0.5em minus 0.4em\relax ACM, 2017, pp. 315--324.

\bibitem{pass2017sleepy}
------, ``The sleepy model of consensus,'' in \emph{International Conference on
  the Theory and Application of Cryptology and Information Security}.\hskip 1em
  plus 0.5em minus 0.4em\relax Springer, 2017, pp. 380--409.

\bibitem{radner1982collusive}
R.~Radner, ``Collusive behavior in noncooperative epsilon-equilibria of
  oligopolies with long but finite lives,'' in \emph{Noncooperative Approaches
  to the Theory of Perfect Competition}.\hskip 1em plus 0.5em minus 0.4em\relax
  Elsevier, 1982, pp. 17--35.

\bibitem{sapirshtein2016optimal}
A.~Sapirshtein, Y.~Sompolinsky, and A.~Zohar, ``Optimal selfish mining
  strategies in bitcoin,'' in \emph{International Conference on Financial
  Cryptography and Data Security}.\hskip 1em plus 0.5em minus 0.4em\relax
  Springer, 2016, pp. 515--532.

\bibitem{syta2017scalable}
E.~Syta, P.~Jovanovic, E.~K. Kogias, N.~Gailly, L.~Gasser, I.~Khoffi, M.~J.
  Fischer, and B.~Ford, ``Scalable bias-resistant distributed randomness,'' in
  \emph{Security and Privacy (SP), 2017 IEEE Symposium on}.\hskip 1em plus
  0.5em minus 0.4em\relax Ieee, 2017, pp. 444--460.

\bibitem{szabo1996smart}
N.~Szabo, ``Smart contracts: building blocks for digital markets,''
  \emph{EXTROPY: The Journal of Transhumanist Thought,(16)}, 1996.

\bibitem{teutsch2017scalable}
J.~Teutsch and C.~Reitwie{\ss}ner, ``A scalable verification solution for
  blockchains,'' 2017.

\bibitem{van2014versum}
J.~van~den Hooff, M.~F. Kaashoek, and N.~Zeldovich, ``Versum: Verifiable
  computations over large public logs,'' in \emph{Proceedings of the 2014 ACM
  SIGSAC Conference on Computer and Communications Security}.\hskip 1em plus
  0.5em minus 0.4em\relax ACM, 2014, pp. 1304--1316.

\bibitem{wald1973sequential}
A.~Wald, \emph{Sequential analysis}.\hskip 1em plus 0.5em minus 0.4em\relax
  Courier Corporation, 1973.

\bibitem{wang2014learning}
J.~Wang, Y.~Song, T.~Leung, C.~Rosenberg, J.~Wang, J.~Philbin, B.~Chen, and
  Y.~Wu, ``Learning fine-grained image similarity with deep ranking,'' in
  \emph{Proceedings of the IEEE Conference on Computer Vision and Pattern
  Recognition}, 2014, pp. 1386--1393.

\bibitem{wang1993number}
Y.~H. Wang, ``On the number of successes in independent trials,''
  \emph{Statistica Sinica}, pp. 295--312, 1993.

\bibitem{zamani2018rapidchain}
M.~Zamani, M.~Movahedi, and M.~Raykova, ``Rapidchain: Scaling blockchain via
  full sharding,'' in \emph{Proceedings of the 2018 ACM SIGSAC Conference on
  Computer and Communications Security}.\hskip 1em plus 0.5em minus 0.4em\relax
  ACM, 2018, pp. 931--948.

\bibitem{zhang2016town}
F.~Zhang, E.~Cecchetti, K.~Croman, A.~Juels, and E.~Shi, ``Town crier: An
  authenticated data feed for smart contracts,'' in \emph{Proceedings of the
  2016 aCM sIGSAC conference on computer and communications security}.\hskip
  1em plus 0.5em minus 0.4em\relax ACM, 2016, pp. 270--282.

\end{thebibliography}
% \input{sections/Appendix}

% that's all folks
\end{document}